\documentclass{article}[11pt]
\usepackage{geometry} 
\usepackage{graphicx}
\usepackage{amsfonts}
\usepackage{amssymb}
\usepackage{amsmath,amsthm}
\usepackage{dsfont,mathtools}
\usepackage{relsize}
\usepackage{wrapfig}
\usepackage{frame,color}
\usepackage{hyperref}
\usepackage{xspace}
\usepackage{framed}
\usepackage{mathrsfs}
\usepackage{mathabx}

\newtheorem{remark}{Remark}
\newtheorem{theorem}{Theorem}
\newtheorem{lemma}{Lemma}
\newtheorem{conjecture}{Conjecture}

\newtheorem{definition}{Definition}
\newtheorem{corollary}{Corollary}

\newcommand{\RandLOCAL}{\mathsf{RandLOCAL}}
\newcommand{\DetLOCAL}{\mathsf{DetLOCAL}}
\newcommand{\ID}{\operatorname{ID}}
\newcommand{\LOCAL}{\mathsf{LOCAL}}
\newcommand{\SLOCAL}{\mathsf{SLOCAL}}

\newcommand{\ignore}[1]{}

\newcommand{\Prob}{\operatorname{Pr}}
\newcommand{\bydef}{\stackrel{\rm def}{=}}

\newcommand{\ceil}[1]{\left\lceil #1 \right\rceil}
\newcommand{\floor}[1]{\lfloor #1 \rfloor}
\newcommand{\f}[2]{\frac{#1}{#2}}

\newcommand{\poly}{{\operatorname{poly}}}
\newcommand{\polylog}{{\operatorname{polylog}}}

\newcommand{\bottom}{\perp}

\newcommand{\Lovasz}{Lov\'{a}sz}
\newcommand{\Goos}{G\"{o}\"{o}s}

\newcommand{\LabelIn}{\Sigma_{\operatorname{in}}}
\newcommand{\LabelOut}{\Sigma_{\operatorname{out}}}
\newcommand{\Lpump}{ {\ell_{\operatorname{pump}}} }

\newcommand{\type}{\textsf{Type}}
\newcommand{\class}{\textsf{Class}}
\newcommand{\replace}{\textsf{Replace}}
\newcommand{\extend}{\textsf{Extend}}
\newcommand{\labeling}{\textsf{Label}}
\newcommand{\pump}{\textsf{Pump}}
\newcommand{\cut}{\textsf{Duplicate}\mbox{-}\textsf{Cut}}

\newcommand{\GG}{\mathcal{G}}
\newcommand{\HH}{\mathcal{H}}
\newcommand{\QQ}{\mathcal{Q}}
\newcommand{\TT}{\mathcal{T}}
\newcommand{\LL}{\mathcal{L}}
\newcommand{\XX}{\mathcal{X}}
\newcommand{\YY}{\mathcal{Y}}
\newcommand{\ZZ}{\mathcal{Z}}

\newcommand{\Pset}{\mathscr{P}}
\newcommand{\Hset}{\mathscr{H}}
\newcommand{\Tset}{\mathscr{T}}

\newcommand{\Gd}[1]{{G}_{#1}}

\newcommand{\simm}{\overset{\star}{\sim}}

\title{A Time Hierarchy Theorem for the LOCAL Model\thanks{Supported by NSF Grants CCF-1514383 and CCF-1637546.}}

\author{Yi-Jun Chang\\ University of Michigan \\ \footnotesize \texttt{cyijun@umich.edu}\and
Seth Pettie\\ University of Michigan \\ \footnotesize \texttt{pettie@umich.edu}}

\begin{document}
\date{}
\maketitle
\thispagestyle{empty}
\setcounter{page}{0}

\begin{abstract}
The celebrated \emph{Time Hierarchy Theorem} for Turing machines states, informally,
that more problems can be solved given more time.  The extent to which a time hierarchy-type theorem
holds in the classic distributed $\LOCAL$ model has been open for many years.  In particular, it is consistent
with previous results that all natural problems in the $\LOCAL$ model can be classified according
to a small \emph{constant} number of complexities, such as $O(1),O(\log^* n), O(\log n), 2^{O(\sqrt{\log n})}$, etc.

In this paper we establish the first time hierarchy theorem for the $\LOCAL$ model and 
prove that several \emph{gaps} exist in the $\LOCAL$ time hierarchy.  Our main results are as follows:
\begin{itemize}
\item We define an infinite set of simple coloring problems called \emph{Hierarchical $2\f{1}{2}$-Coloring}.  A correctly
colored graph can be confirmed by simply checking the neighborhood of each vertex, so this problem fits
into the class of \emph{locally checkable labeling} (LCL) problems.  However, the complexity of the $k$-level
Hierarchical $2\f{1}{2}$-Coloring problem is $\Theta(n^{1/k})$, for $k\in\mathbb{Z}^+$.
The upper and lower bounds hold for both general graphs and trees, and for both randomized and deterministic algorithms.

\item Consider any LCL problem on \emph{bounded degree trees}.  We prove an automatic-speedup theorem
that states that any \emph{randomized} $n^{o(1)}$-time algorithm solving the LCL
can be transformed into a \emph{deterministic} $O(\log n)$-time algorithm.
Together with a previous result~\cite{ChangKP16},
this establishes that on trees, there are no natural deterministic complexities in the ranges
$\omega(\log^* n)$---$o(\log n)$
or $\omega(\log n)$---$n^{o(1)}$.

\item We expose a gap in the \emph{randomized} time hierarchy on general graphs.  Roughly speaking,
any randomized algorithm that solves an LCL problem in sublogarithmic time can be 
sped up to run in $O(T_{LLL})$ time, which is the complexity of the distributed \Lovasz{} local lemma problem, currently known to be $\Omega(\log\log n)$ and $O(\log n)$.
\end{itemize}

Finally, we revisit Naor and Stockmeyer's characterization of $O(1)$-time $\LOCAL$ algorithms for LCL problems 
(as \emph{order-invariant w.r.t.~vertex IDs}) and calculate the complexity gaps that are directly implied 
by their proof.   For $n$-rings we see a $\omega(1)$---$o(\log^* n)$ complexity gap,
for $(\sqrt{n}\times \sqrt{n})$-tori an $\omega(1)$---$o(\sqrt{\log^* n})$ gap, 
and for bounded degree trees and general graphs, an $\omega(1)$---$o(\log(\log^* n))$ complexity gap.
\end{abstract}

\newpage

\section{Introduction}\label{sect:intro}

The goal of this paper is to understand the spectrum of natural problem complexities that can exist in the
$\LOCAL$ model~\cite{Linial92,Peleg00}
of distributed computation, and to quantify the value of randomness in this model.
Whereas the time hierarchy of Turing machines is known\footnote{For \underline{\emph{any}} time-constructible
function $T(n)$, there is a problem solvable in $O(T(n))$ but not $o(T(n))$ time~\cite{HartmanisS65,Furer84}.}
to be very ``smooth'',
recent work~\cite{ChangKP16,Brandt+17} has exhibited strange
\emph{gaps} in the $\LOCAL$ complexity hierarchy.
Indeed, prior to this work it was not even known
if the $\LOCAL$ model could support more than a small \emph{constant} number of problem complexities.
Before surveying prior work in this area, let us formally define the deterministic and randomized variants of the $\LOCAL$ model,
and the class of \emph{locally checkable labeling} (LCL) problems, which are intuitively those graph problems
that can be computed locally in \emph{nondeterministic constant time}.

In both the $\DetLOCAL$ and $\RandLOCAL$ models the input graph $G=(V,E)$ and communications network are identical.
Each vertex hosts a processor and all vertices run the same algorithm.  Each edge supports communication in both directions.
The computation proceeds in synchronized {\em rounds}.  In a round, each processor
performs some computation and sends a message along each incident edge, which is delivered before the beginning
of the next round.  Each vertex $v$ is initially aware of its degree $\deg(v)$, a port numbering mapping its incident edges
to $\{1,\ldots,\deg(v)\}$, certain global parameters such as $n \bydef |V|$,
$\Delta  \bydef \max_{v\in V} \deg(v)$, and possibly other information. The assumption that global parameters are common knowledge can sometimes be removed; see Korman, Sereni, and Viennot~\cite{KormanSV13}.
The only measure of efficiency is the number of rounds.  All local computation is free and the size of messages is unbounded.
\emph{Henceforth ``time'' refers to the number of rounds.}  The differences between $\DetLOCAL$ and $\RandLOCAL$ are as follows.

\begin{description}
\item[$\DetLOCAL$:] In order to avoid trivial impossibilities, all vertices are assumed to hold unique $\Theta(\log n)$-bit IDs.
Except for the information about $\deg(v), \ID(v)$, and the port numbering,
the initial state of $v$ is identical to every other vertex.
The algorithm executed at each vertex is deterministic.

\item[$\RandLOCAL$:] In this model each vertex may locally generate an unbounded number of independent truly random bits.
(There are no globally shared random bits.)  Except for the information about $\deg(v)$ and its port numbering,
the initial state of $v$ is identical to
every other vertex.  Algorithms in this model operate for a specified number of rounds and have some probability of {\em failure},
the definition of which is problem specific.
We set the maximum tolerable global probability of failure to be $1/n$.
\end{description}

Clearly $\RandLOCAL$ algorithms can generate distinct IDs (w.h.p.) if desired.
Observe that the role of ``$n$'' is different in the two $\LOCAL$ models:
in $\DetLOCAL$ it affects the ID length whereas in $\RandLOCAL$ it affects the failure probability.

\paragraph{LCL Problems.} Naor and Stockmeyer~\cite{NaorS95} introduced \emph{locally checkable labelings}
to formalize a large class of natural graph problems.
Fix a class $\mathcal{G}$ of possible input graphs and let $\Delta$ be the maximum degree in any such graph.
Formally, an LCL problem $\mathcal{P}$ for $\mathcal{G}$ has a radius $r$, constant size input and output alphabets $\LabelIn,\LabelOut$,
and a set $\mathcal{C}$ of acceptable configurations. All of these parameters may depend on $\Delta$.
Each $C\in \mathcal{C}$ is a graph
centered at a specific vertex, in which each vertex has
 a degree, a port numbering,
and
two labels from $\LabelIn$ and $\LabelOut$.
Given the input graph $G(V,E,\phi_{\operatorname{in}})$ where $\phi_{\operatorname{in}} : V(G) \rightarrow \LabelIn$,
an acceptable output is any function $\phi_{\operatorname{out}}  : V(G) \rightarrow \LabelOut$ such that
for each $v\in V(G)$, the subgraph induced by $N^r(v)$
(denoting the $r$-neighborhood of $v$ together with information stored there:
vertex degrees,
port numberings,
input labels, and output labels)
is isomorphic to a member of $\mathcal{C}$.
An LCL can be described explicitly by enumerating a finite number of acceptable configurations.
LCLs can be generalized to graph classes with unbounded degrees.

Many natural symmetry breaking problems can be expressed as LCLs, such as
MIS, maximal matching, $(\alpha,\beta)$-ruling sets, $(\Delta+1)$-vertex coloring, and sinkless orientation.

\subsection{The Complexity Landscape of $\LOCAL$}

The complexity landscape for LCL problems is defined by ``natural'' complexities (sharp lower and upper bounds
for specific LCL problems) and provably empty \emph{gaps} in the complexity spectrum.  We now have an almost perfect understanding
of the complexity landscape for two simple topologies: $n$-rings~\cite{ColeV86,Linial92,Naor91,NaorS95,ChangKP16}
and $(\sqrt{n}\times \sqrt{n})$-tori~\cite{NaorS95,ChangKP16,Brandt+17}.  See Figure~\ref{fig:complexity}, Top and Middle.
On the $n$-ring, the only possible problem complexities are
$O(1)$, $\Theta(\log^* n)$ (e.g., 3-coloring), and $\Theta(n)$ (e.g., 2-coloring, if bipartite).
The gaps between these three complexities are obtained by \emph{automatic speedup theorems}.
Naor and Stockmeyer's~\cite{NaorS95} characterization of $O(1)$-time LCL algorithms actually implies
that any $o(\log^* n)$-time algorithm on the $n$-ring can be transformed to run in $O(1)$ time; see Appendix~\ref{sect:NaorStockmeyer}.
Chang, Kopelowitz, and Pettie~\cite{ChangKP16}
showed that any $o(n)$-time $\RandLOCAL$
algorithm can be made to run in $O(\log^* n)$ time in $\DetLOCAL$.

The situation with $(\sqrt{n}\times \sqrt{n})$-tori is almost identical~\cite{Brandt+17}:
every known LCL has complexity $O(1)$, $\Theta(\log^* n)$ (e.g., 4-coloring),
or $\Theta(\sqrt{n})$ (e.g., 3-coloring).
Whereas the gap implied by~\cite{NaorS95} is $\omega(1)$---$o(\log^* n)$
on the $n$-ring, it is $\omega(1)$---$o(\sqrt{\log^* n})$ on the $(\sqrt{n}\times \sqrt{n})$-torus; see Appendix~\ref{sect:NaorStockmeyer}.\footnote{J. Suomela (personal communication, 2017) has a proof that
there is an $\omega(1)$---$o(\log^\ast n)$ complexity gap for tori, at least for LCLs that do not use port
numberings or input labels.  The issues that arise with port numbering and input labels can be very subtle.}
Whereas randomness is known not to help in $n$-rings~\cite{NaorS95,ChangKP16},
it is an open question on tori~\cite{Brandt+17}.
Whereas the classification question is decidable on $n$-rings
(whether an LCL is $O(\log^* n)$ or $\Omega(n)$, for example)
this question is \emph{undecidable} on $(\sqrt{n}\times \sqrt{n})$-tori~\cite{NaorS95,Brandt+17}.

The gap theorems of Chang et al.~\cite{ChangKP16} show that no LCL problem on general graphs has
$\DetLOCAL$ complexity in the range $\omega(\log^* n)$---$o(\log_\Delta n)$, nor $\RandLOCAL$ complexity
in the range $\omega(\log^* n)$---$o(\log_\Delta\log n)$.
Some problems exhibit an exponential separation ($O(\log_\Delta\log n)$ vs. $\Omega(\log_\Delta n)$)
between their $\RandLOCAL$ and $\DetLOCAL$ complexities, such as $\Delta$-coloring degree-$\Delta$ trees~\cite{BrandtEtal16,ChangKP16} and sinkless orientation~\cite{BrandtEtal16,GhaffariS17}.
More generally, Chang et al.~\cite{ChangKP16} proved that the $\RandLOCAL$ complexity
of \emph{any} LCL problem on graphs of size $n$ is, holding $\Delta$ fixed,
at least its deterministic complexity on instances of size $\sqrt{\log n}$.
Thus, on the class of degree $\Delta=O(1)$ graphs there were only five known natural complexities:
$O(1)$, $\Theta(\log^* n)$, randomized $\Theta(\log\log n)$,
$\Theta(\log n)$, and $\Theta(n)$.
For non-constant $\Delta$, the $\RandLOCAL$ lower bounds of
Kuhn, Moscibroda, and Wattenhofer~\cite{KuhnMW16} imply
$\Omega(\min\{\f{\log\Delta}{\log\log\Delta}, \sqrt{\f{\log n}{\log\log n}}\})$ lower bounds on $O(1)$-approximate
vertex cover, MIS, and maximal matching.  This $\Omega(\log\Delta/\log\log\Delta)$
lower bound is only known to be tight for $O(1)$-approximate vertex cover~\cite{Bar-YehudaCS16};
the best maximal matching~\cite{BEPS16} and
MIS~\cite{Ghaffari16} algorithms' dependence on $\Delta$ is $\Omega(\log \Delta)$.
The $\Omega(\sqrt{\f{\log n}{\log\log n}})$ lower bound is not known to be tight for any problem, but is almost
tight for maximal matching on bounded arboricity graphs~\cite{BEPS16}, e.g., trees or planar graphs.

\begin{figure}
\centering
\scalebox{.4}{\includegraphics{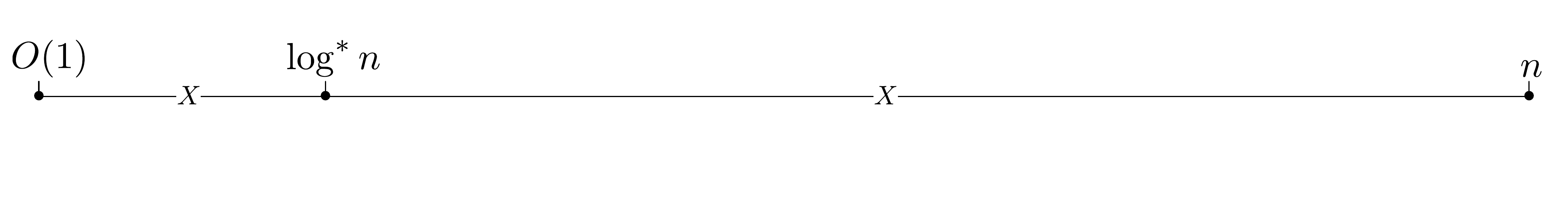}}

\scalebox{.4}{\includegraphics{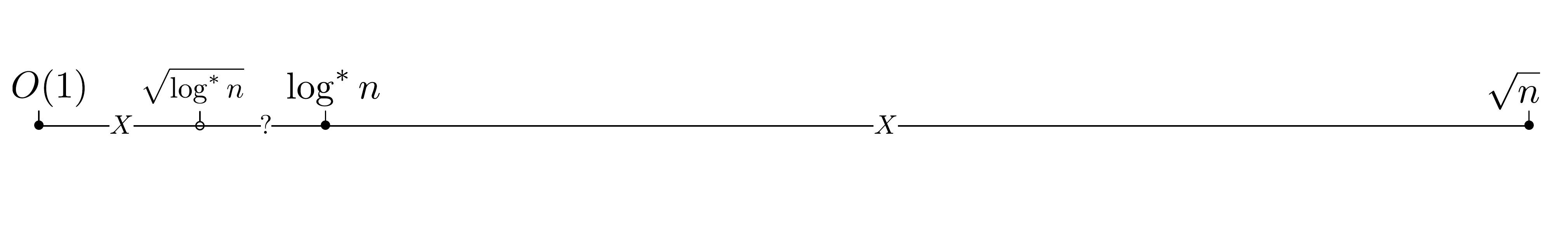}}

\scalebox{.4}{\includegraphics{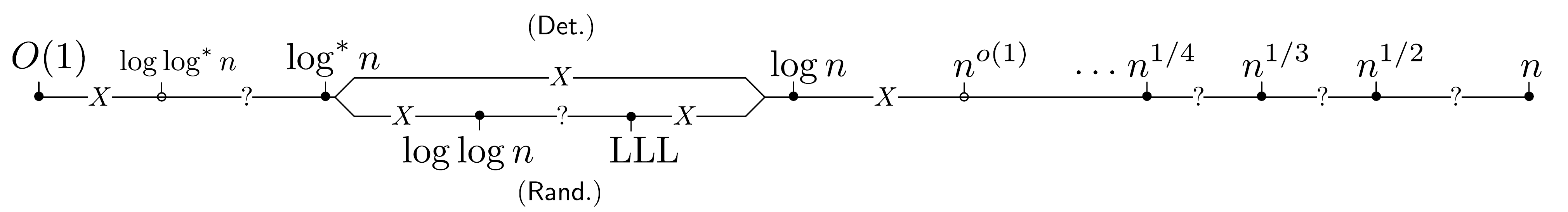}}
\caption{\label{fig:complexity}%
{\bf Top:} the complexity landscape for LCL problems on the $n$-ring.
{\bf Middle:} the complexity landscape for LCL problems on the $(\sqrt{n}\times \sqrt{n})$-torus.
Refer to \cite{NaorS95,ChangKP16,Brandt+17} and Appendix~\ref{sect:NaorStockmeyer}
for proofs of the complexity gaps (`$X$') on rings and tori.
{\bf Bottom:} the complexity landscape for LCL problems on bounded degree trees.
The $\omega(\log^* n)$---$o(\log n)$ $\DetLOCAL$ gap and $\omega(\log^* n)$---$o(\log\log n)$ $\RandLOCAL$ gap
are due to~\cite{ChangKP16}.
The $\omega(LLL)$---$o(\log n)$ and $\omega(\log n)$---$n^{o(1)}$ gaps are new.
Refer to Appendix~\ref{sect:NaorStockmeyer} for the $\omega(1)$---$o(\log(\log^* n))$ gap.
It is unknown whether
there are $\omega(n^{1/(k+1)})$---$o(n^{1/k})$ gaps.
With the exception of the $\omega(\log n)$---$n^{o(1)}$ gap, this is exactly the
known complexity landscape for general bounded degree graphs as well.}
\end{figure}

\paragraph{New Results.}
In this paper we study the $\LOCAL$ complexity landscape on more general topologies:
bounded degree trees and general graphs; see Figure~\ref{fig:complexity}, Bottom.
We establish a new complexity gap for trees, a \emph{potential} complexity gap for general graphs
(which depends on the distributed complexity of the constructive \Lovasz{} local lemma),
and a new infinite hierarchy of coloring problems with polynomial time complexities.
In more detail,

\begin{itemize}
\item We prove that on the class of degree bounded trees, no LCL has complexity in the range
$\omega(\log n)$---$n^{o(1)}$.  Specifically, any $n^{o(1)}$-time $\RandLOCAL$ algorithm can
be converted to an $O(\log n)$-time $\DetLOCAL$ algorithm.  Moreover, given a description of an LCL problem $\mathcal{P}$,
it is {\em decidable} whether the $\RandLOCAL$ complexity of $\mathcal{P}$ is
$n^{\Omega(1)}$ or the $\DetLOCAL$ complexity of $\mathcal{P}$ is $O(\log n)$.

It turns out that this gap is maximal:
we cannot extend it lower than $\omega(\log n)$~\cite{Linial92,ChangKP16}, nor higher than $n^{o(1)}$,
as we show below.

\item We define an infinite class of LCL problems called \emph{Hierarchical $2\f{1}{2}$-Coloring}.
We prove that $k$-level Hierarchical $2\f{1}{2}$-Coloring has complexity $\Theta(n^{1/k})$.
The upper bound holds in $\DetLOCAL$ on general graphs, and the lower bound holds
even on degree-3 trees in $\RandLOCAL$.  Thus, in contrast to rings and tori,
trees and general graphs support an \emph{infinite} number of natural problem complexities.

\item Suppose we have a $\RandLOCAL$ algorithm for general graphs running in
$C(\Delta) + o(\log_\Delta n)$ time.  We can transform this algorithm to run in
$O(C(\Delta)\cdot T_{LLL})$ time, where $T_{LLL}$ is the complexity
of a weak (i.e., ``easy'') version of the constructive \Lovasz{} local lemma.  At present,
$T_{LLL}$ is known to be $\Omega(\log\log n)$ and $O(\log n)$ on bounded degree graphs.
If it turns out that $T_{LLL}$ is sublogarithmic, this establishes a new $\RandLOCAL$ complexity
gap between $\omega(T_{LLL})$ and $o(\log n)$ on bounded degree graphs.
\end{itemize}

Finally, it seems to be folklore that Naor and Stockmeyer's work~\cite{NaorS95} implies \emph{some kind}
of complexity gap, which has been cited as $\omega(1)$---$o(\log^* n)$~\cite[p. 2]{Brandt+17}.
However, to our knowledge, no proof of this complexity gap has been published.
We show how Naor and Stockmeyer's approach implies complexity gaps that depend on the graph topology:
\begin{itemize}
\item[---] $\omega(1)$---$o(\log^* n)$ on rings.
\item[---] $\omega(1)$---$o(\sqrt{\log^* n})$ on tori.
\item[---] $\omega(1)$---$o(\log(\log^* n))$ on bounded degree trees and general graphs.
\end{itemize}
These gaps apply to the general class of LCL problems defined in this paper, in which vertices
initially hold an input label and possible port numbering. Port numberings are needed to
represent ``edge labeling'' problems (like maximal matching, edge coloring, and sinkless orientation)
unambiguously as vertex labelings.  They are not needed for native ``vertex labeling'' problems
like $(\Delta+1)$-coloring or MIS.  J. Suomela (personal communication) gave a proof that the
$\omega(1)$---$o(\log^\ast n)$ gap exists in tori as well, for the class of LCL problems without input labels
or port numbering; see Appendix~\ref{sect:NaorStockmeyer}.

\paragraph{Commentary.}
Our $\omega(\log n)$---$n^{o(1)}$ complexity gap for trees is interesting from both a technical
and greater philosophical perspective, due to the fact that many natural problems have been
``stuck'' at $n^{o(1)}$ complexities for decades.  Any $\DetLOCAL$ algorithm that relies
on network decompositions~\cite{PanconesiS96} currently takes $2^{O(\sqrt{\log n})}$ time.
If our automatic speedup theorem could be extended to the class of all graphs,
this would immediately yield $O(\log n)$-time algorithms for MIS, $(\Delta+1)$-coloring, and many other LCLs.

All the existing automatic speedup theorems are quite different in terms of proof techniques.
Naor and Stockmeyer's approach is based on Ramsey theory.
The speedup theorems of~\cite{ChangKP16,Brandt+17}
use the fact that $o(\log_\Delta n)$ algorithms on general graphs (and $o(n)$ algorithms
on $n$-rings and $o(\sqrt{n})$ algorithms on $(\sqrt{n}\times \sqrt{n})$-tori)
cannot ``see'' the whole graph, and can therefore
be efficiently tricked into thinking the graph has constant size.
Our $n^{o(1)} \rightarrow O(\log n)$ speedup theorem
introduces an entirely new set of techniques based on classic automata theory.
We show that any LCL
problem gives rise to a regular language that represents partial labelings of the tree that
can be consistently extended to total lablelings.
By applying the pumping lemma for regular languages,
we can ``pump'' the input tree into a much larger tree that behaves similar to the original tree.
The advantage of creating a larger \underline{imaginary} tree is that each vertex can
(mentally) simulate the behavior of an $n^{o(1)}$-time algorithm on the \emph{imaginary} tree,
merely by inspecting its $O(\log n)$-neighborhood in the \emph{actual} tree.
Moreover, because the pumping operation preserves properties of the original tree, a labeling
of the imaginary tree can be efficiently converted to a labeling of the original tree.

\subsection{Related Results}

There are several $\LOCAL$ lower bounds for natural problems that do not quite fit in the LCL framework.
\Goos, Hirvonen, and Suomela~\cite{GoosHS15} proved a sharp $\Omega(\Delta)$
lower bound for fractional maximal matching and \Goos{} and Suomela proved $\Omega(\log n)$ lower bounds on $(1+\delta)$-approximating the minimum vertex cover, $\delta>0$, even on degree-3 graphs.
See~\cite{KuhnW06,HefetzKMS16} for lower bounds on coloring problems
that apply to constrained algorithms or a constrained version of the $\LOCAL$ model.

In recent years there have been efforts to develop a `complexity
theory' of locality.  The gap theorems of~\cite{NaorS95,ChangKP16,Brandt+17} have already been discussed.
Suomela surveys~\cite{Suomela13} the class of problems that can be computed with $O(1)$ time.
Fraigniaud et al.~\cite{FraigniaudKP13} defined a distributed model for locally
deciding graph properties; see~\cite{FeuilloleyF16} for a survey of variants of the local distributed decision model.
\Goos{} and Suomela~\cite{GoosS16} considered the \emph{proof} complexity (measured in terms of bits-per-vertex label) of locally verifying graph properties.  Very recently, Ghaffari, Kuhn, and Maus~\cite{GhaffariKM16} defined
the $\SLOCAL$ model (sequential $\LOCAL$) and exhibited several \emph{complete} problems for this model,
inasmuch as a $\polylog(n)$-time $\DetLOCAL$ algorithm for any complete problem implies a $\polylog(n)$
$\DetLOCAL$ algorithm for every $\polylog(n)$-time problem in $\SLOCAL$.\footnote{The class of $O(1)$-time $\SLOCAL$ algorithms is, roughly speaking, those graph labelings
that can be computed sequentially, by a truly local algorithm.
This class is a \emph{strict} subset of LCLs.}

\subsection{Organization}

In Section~\ref{sec.poly} we introduce Hierarchical $2\f{1}{2}$-Coloring and prove that the $k$-level variant
of this problem has complexity $\Theta(n^{1/k})$.
In Section~\ref{sec.tree-gap} we prove the $n^{o(1)}\rightarrow O(\log n)$ speedup theorem for bounded degree trees.
In Section~\ref{sect:LLL} we discuss the constructive \Lovasz{} local lemma and prove
the $o(\log_\Delta n) \rightarrow T_{LLL}$ randomized speedup theorem.
In Section~\ref{sect:conclusion} we discuss open problems and outstanding conjectures.
Appendix~\ref{sect:NaorStockmeyer} reviews Naor and Stockmeyer's characterization
of $O(1)$-time LCL algorithms, using Ramsey theory, and explains how it implies gaps
in the complexity hierarchy that depend on graph topology.

\section{An Infinitude of Complexities: Hierarchical $2\f{1}{2}$-Coloring\label{sec.poly}}

\newcommand{\bl}{\Venus}
\newcommand{\wh}{\Mars}
\newcommand{\none}{\Mercury}
\newcommand{\exempt}{\Saturn}

In this section we give an infinite sequence $(\mathcal{P}_k)_{k\in \mathbb{Z}^+}$ of LCL problems, where the complexity
of $\mathcal{P}_k$ is precisely $\Theta(n^{1/k})$.\footnote{Brandt et al.~\cite[Appendix A.3]{Brandt+17} described an LCL that has
complexity $\Theta(\sqrt{n})$ on general graphs, but not trees.
It may be possible to generalize their LCL to any complexity of the form $\Theta(n^{1/k})$.}
The upper bound holds on general graphs in $\DetLOCAL$
and the lower bound holds in $\RandLOCAL$, even on degree-3 trees.
Informally, the task of $\mathcal{P}_k$ is to 2-color (with $\{\bl,\wh\}$) certain specific subgraphs of the input graph.
Some vertices are \emph{exempt} from being colored (in which case they are labeled $\exempt$), and in addition,
it is possible to decline to 2-color certain subgraphs, by labeling them $\none$.

There are no input labels.  The output label set is $\LabelOut = \{\bl,\wh,\none,\exempt\}$.\footnote{Venus, Mars, Mercury, Saturn.}
The problem $\mathcal{P}_k$ is an LCL defined by the following rules.
\begin{description}
\item[Levels.] Subsequent rules depend on the {\em levels} of vertices.  Let $V_i$, $i\in\{1,\ldots,k+1\}$, be the
set of vertices on level $i$, defined as follows.
\begin{align*}
G_1 &= G\\
G_i &= G_{i-1} - V_{i-1}, 					&& \mbox{ for } i\in [2,k+1]\\
V_i &= \{v \in V(G_i) \;|\; \deg_{G_i}(v) \le 2\},  && \mbox{ for } i\in [1,k]\\
V_{k+1} &= V(G_{k+1})					 && \mbox{ (the remaining vertices)}
\end{align*}
Remember that vertices know their degrees, so a vertex in $V_1$ deduces this with 0 rounds of communication.  In general
the level of $v$ can be calculated from information in $N^k(v)$.

\item[Exemption.] A vertex labeled $\exempt$ is called \emph{exempt}.
No $V_1$ vertex is labeled $\exempt$; all $V_{k+1}$ vertices are labeled $\exempt$.
Any $V_i$ vertex is labeled $\exempt$ iff it is adjacent to a lower level vertex labeled $\bl,\wh,$ or $\exempt$.
Define $D_i \subseteq V_i$ to be the set of level $i$ exempt vertices.

\item[Two-Coloring.] Vertices not covered by the exemption rule are labeled one of $\bl,\wh,\none$.
\begin{itemize}
\item[---] Any vertex in $V_i$, $i\in[1,k]$, labeled $\bl$ has no neighbor in $V_i$ labeled $\bl$ or $\none$.
\item[---] Any vertex in $V_i$, $i\in[1,k]$, labeled $\wh$ has no neighbor in $V_i$ labeled $\wh$ or $\none$.
\item[---] Any vertex in $V_k-D_k$ with exactly 0 or 1 neighbors in $V_k-D_k$ must be labeled $\bl$ or $\wh$.
\end{itemize}
\end{description}

\paragraph{Commentary.}
The Level rule states that the graph induced by $V_i$ consists of paths and cycles.
The Two-Coloring rule implies that each component of non-exempt vertices in the graph induced by $V_i - D_i$
must either
(a) be labeled uniformly by $\none$ or
(b) be properly 2-colored by $\{\bl,\wh\}$.
Every \emph{path} in  $V_k - D_k$ must be properly 2-colored, but \emph{cycles} in $V_k - D_k$ are allowed to be labeled
uniformly by $\none$.  This last provision is necessary to ensure that \emph{every} graph can be labeled according
to $\mathcal{P}_k$ since there is no guarantee that cycles in $V_k - D_k$ are bipartite.

\begin{remark}
As stated $\mathcal{P}_k$ is an LCL with an alphabet size of 4 and a radius $k$, since
the coloring rules refer to levels, which can be deduced by looking up to radius $k$.
On the other hand, we can also represent $\mathcal{P}_k$ as an LCL with radius 1
and alphabet size $4k$ by including a vertex's level in its output label.  A correct level assignment
can be verified within radius 1.  For example, level 1 vertices are those with degree at most 2, and
a vertex is labeled $i\in[2,k]$ iff all but at most 2 neighbors have levels less than $i$.
\end{remark}

\begin{theorem}
The $\DetLOCAL$ complexity of $\mathcal{P}_k$ on general graphs is $O(n^{1/k})$.
\end{theorem}

\begin{proof}
The algorithm fixes the labeling of $V_1,\ldots,V_k,V_{k+1}$ in order, according to the following steps.
Assume that all vertices in $V_1,\ldots,V_{i-1}$ have already been labeled.
\begin{itemize}
\item Compute $D_i$ according to the Exemption rule.  (E.g., $D_1 = \emptyset$, $D_{k+1} = V_{k+1}$.)

\item Each path in the subgraph induced by $V_i - D_i$ calculates its length.
If it contains at most $\ceil{2n^{1/k}}$ vertices, it
properly 2-colors itself with $\{\bl,\wh\}$; longer paths and cycles in $V_i - D_i$ label themselves uniformly by $\none$.
\end{itemize}

This algorithm correctly solves $\mathcal{P}_k$ provided that it never labels a path in $V_k - D_k$ with $\none$.
Let $U_i$ be the subgraph induced by those vertices in $V_1\cup\cdots\cup V_i$ labeled $\none$.
Consider a connected component $C$ in $U_i$ whose $V_i$-vertices are arranged in a path (not a cycle).
We argue by induction that $C$ has at least $2n^{i/k}$ vertices.
This is clearly true in the base case $i=1$: if a path component of $U_1$ were colored $\none$, it must
have more than $\ceil{2n^{1/k}}$ vertices.  Now assume the claim is true for $i-1$ and consider a component
$C$ of $U_i$.  If the $V_i$-vertices in $C$ form a path, it must have length greater than $2n^{1/k}$.
Each vertex in that path must be adjacent to an endpoint of a $V_{i-1}$ path.  Since $V_{i-1}$ paths have two endpoints,
the $V_i$ path is adjacent to at least $\ceil{2n^{1/k}}/2 \ge n^{1/k}$ components in $U_{i-1}$, each of which has size
at least $2n^{(i-1)/k}$, by the inductive hypothesis.  Thus, the size of $C$ is at least $n^{1/k}\cdot 2n^{(i-1)/k} + 2n^{1/k} > 2n^{i/k}$.
Because there are at most $n$ vertices in the graph, it is impossible for $V_k$ vertices arranged in a path
to be colored $\none$.
\end{proof}

\begin{theorem}\label{thm:LCL-lb}
The $\RandLOCAL$ complexity of $\mathcal{P}_k$ on trees with maximum degree $\Delta = 3$ is $\Omega(n^{1/k})$.
\end{theorem}

\begin{proof}
Fix an integer parameter $x$ and define a sequence of graphs $(H_i)_{1\le i\le k}$ as follows.
Each $H_i$ has a \emph{head} and a \emph{tail}.
\begin{itemize}
\item $H_1$ is a path (or \emph{backbone}) of length $x$.  One end of the path is the head and the other end the tail.

\item To construct $H_i$, $i\in [2,k-1]$, begin with a \emph{backbone} path $(v_1,v_2,\ldots,v_x)$,
with head $v_1$ and tail $v_x$.
Form $x+1$ copies $(H_{i-1}^{(j)})_{1\le j\le x+1}$ of $H_{i-1}$, where $v^{(j)}$ is the head of $H_{i-1}^{(j)}$.
Connect $v^{(j)}$ to $v_j$ by an edge, for $j\in[1,x]$, and also connect $v^{(x+1)}$ to $v_x$ by an edge.

\item $H_k$ is constructed exactly as above, except that we generate $x+2$ copies of $H_{k-1}$ and connect the heads of
two copies of $H_{k-1}$ to both $v_1$ and $v_x$.   See Figure~\ref{fig:lower-bound-graph} for an example with $k=3$.
\end{itemize}

\begin{figure}
\centering
\scalebox{.45}{\includegraphics{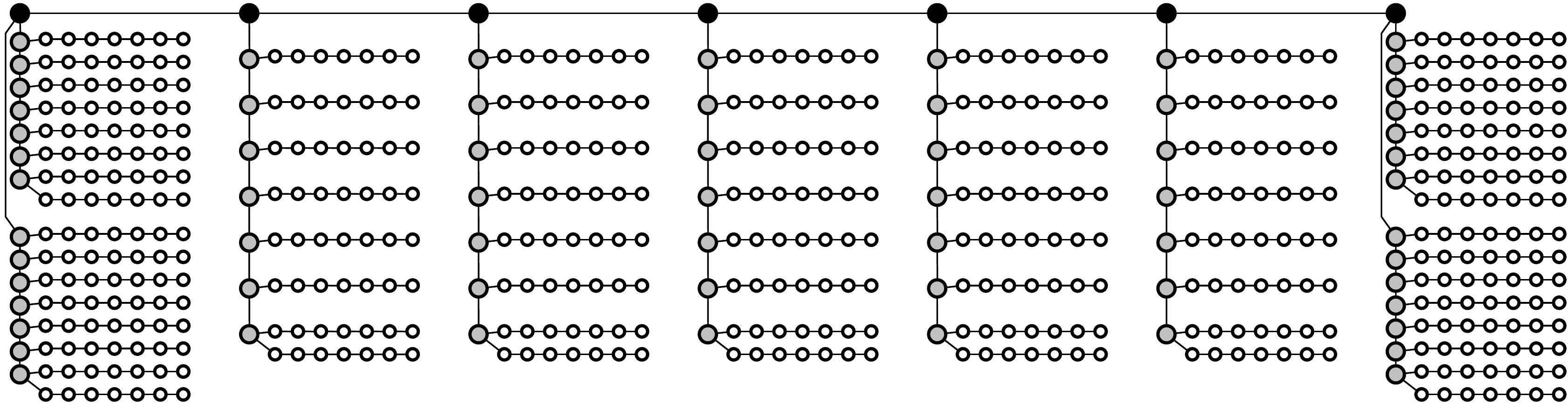}}
\caption{\label{fig:lower-bound-graph}The graph $H_k$ with parameters $k=3, x=7$.  White vertices are in $V_1$, gray in $V_2$,
and black in $V_3$.  $V_4 = V_{k+1}$ is empty.}
\end{figure}

Let us make several observations about the construction of $H_k$.
First, it is a tree with maximum degree 3.
Second, when decomposing $V(H_k)$ into levels $(V_1,\ldots,V_k,V_{k+1})$,
$V_i$ is precisely the union of the backbones in all copies of $H_i$, and $V_{k+1}=\emptyset$.
Third, the number of vertices in
$H_k$ is $\Theta(x^k)$, so a $o(n^{1/k})$ algorithm for $\mathcal{P}_k$ must run in $o(x)$ time on $H_k$.

Consider a $\RandLOCAL$ algorithm $\mathcal{A}$ solving $\mathcal{P}_k$ on $H_k$ within $t < x/5 - O(1)$ time,
that fails with probability $p_{\operatorname{fail}}$.  If $\mathcal{A}$ is a good algorithm then
$p_{\operatorname{fail}} \le 1/|V(H_k)|$.  However, we will now show that $p_{\operatorname{fail}}$ is constant,
independent of $|V(H_k)|$.

Define $\mathcal{E}_i$ to be the event that $D_i\neq \emptyset$ and $p_i = \Pr(\mathcal{E}_i)$.
By an induction from $i=2$ to $k$, we prove that $p_i \leq 2(i-1)\cdot p_{\operatorname{fail}}$.

\paragraph{Base case.}
We first prove that
\[
\Prob\left(\mbox{$H_k$ is not correctly colored according to $\mathcal{P}_k$ }  \;|\; \mathcal{E}_2\right) \geq 1/2.
\]
Conditioning on $\mathcal{E}_2$ means that $D_2\neq\emptyset$.
Fix any $v\in D_2$ and let $P$ be a copy of $H_1$ (a path) adjacent to $v$.
In order for $v\in D_2$, it must be that $P$ is properly 2-colored with $\{\bl,\wh\}$.
Since $t < x/5 - O(1)$, there exist two vertices $u$ and $u'$ in $P$ such that
\begin{enumerate}
\item $N^{t}(u)$, $N^{t}(u')$, and $N^{t}(v)$ are disjoint sets,
\item the subgraphs induced by $N^{t}(u)$ and $N^{t}(u')$ are isomorphic, and
\item the distance between $u$ and $u'$ is odd.
\end{enumerate}
Let $p_{\bl}$ and $p_{\wh}$ be the probabilities that $u$/$u'$ is labeled $\bl$ and $\wh$, respectively.
A proper 2-coloring of $P$ assigns $u$ and $u'$ different colors, and that occurs with probability
$2p_{\bl}p_{\wh} \le 2p_{\bl}(1-p_{\bl}) \le 1/2$.
Moreover, this holds independent of the random bits generated by vertices in $N^t(v)$.
The algorithm fails unless $u,u'$ have different colors, thus
$p_{\operatorname{fail}} \ge p_2/2$, and hence
$p_2 \leq 2\cdot p_{\operatorname{fail}}$.

\paragraph{Inductive Step.} Let $3 \leq i \leq k$.
The inductive hypothesis states that $p_{i-1} \leq 2(i-2)\cdot p_{\operatorname{fail}}$.
By a proof similar to the base case, we have that:
\[
\Prob\left(\mbox{$H_k$ is not correctly colored according to $\mathcal{P}_k$ }  \;|\; \mathcal{E}_i\backslash\mathcal{E}_{i-1}\right) \geq 1/2.
\]
We are conditioning on $\mathcal{E}_i\backslash\mathcal{E}_{i-1}$.  If this event is empty,
then $p_i \le p_{i-1} \leq 2(i-2)\cdot p_{\operatorname{fail}}$ and the induction is complete.
On the other hand, if $\mathcal{E}_i\backslash\mathcal{E}_{i-1}$ holds then
there is some $v\in D_i$ adjacent to a copy of $H_{i-1}$ with backbone path $P$, where
$P\cap D_{i-1} = \emptyset$.  In other words, if $H_k$ is colored according to $\mathcal{P}_k$
then $P$ must be properly 2-colored with $\{\bl,\wh\}$.
The argument above shows this occurs with probability at least 1/2.
Thus,
\[
p_{\operatorname{fail}} = \Pr(\mbox{$H_k$ is incorrectly colored}) \ge \Pr(\mathcal{E}_i\backslash\mathcal{E}_{i-1}) / 2 \ge (p_i-p_{i-1})/2,
\]
or $p_i \le 2p_{\operatorname{fail}} + p_{i-1} \le 2(i-1)p_{\operatorname{fail}}$, completing the induction.

Finally, let $P$ be the path induced by vertices in $V_k$.
The probability that $\mathcal{E}_k$ holds ($P\cap D_k \neq \emptyset$) is $p_k \le 2(k-1)\cdot p_{\operatorname{fail}}$.
On the other hand, $\Pr(\mbox{$H_k$ not colored correctly} \;|\; \overline{\mathcal{E}_k}) \ge 1/2$ by the argument above,
hence $p_{\operatorname{fail}} \ge (1-p_k)/2$, or $p_k \ge 1 - 2p_{\operatorname{fail}}$.
Combining the upper and lower bounds on $p_k$ we conclude that $p_{\operatorname{fail}} \ge (2k)^{-1}$ is constant,
independent of $|V(H_k)|$.
Thus, algorithm $\mathcal{A}$ cannot succeed with high probability.
\end{proof}

\section{A Complexity Gap on Bounded Degree Trees\label{sec.tree-gap}}

In this section we prove an $n^{o(1)} \rightarrow O(\log n)$ speedup theorem for LCL
problems on bounded degree trees.  The progression of definitions and lemmas in
Sections~\ref{sect:partial-labeling}--\ref{sect:feasible-functions} is \emph{logical}, but 
obscures the high level structure of the proof.  Section~\ref{sect:tour} gives an informal
tour of the proof and its key ideas.  Throughout, $\mathcal{P}$ is a radius-$r$ LCL
and $\mathcal{A}$ is an $n^{o(1)}$-time algorithm for $\mathcal{P}$ on bounded degree trees.

\subsection{A Tour of the Proof}\label{sect:tour}

Consider this simple way to decompose a tree in $O(\log n)$ time, inspired by Miller and Reif~\cite{MillerR89}.
Iteratively remove paths of degree-2 vertices (\emph{compress})
and vertices with degree 0 or 1 (\emph{rake}).  Vertices removed in iteration $i$ are at \emph{level $i$}.
If $O(\log n)$ \emph{rakes} alone suffice to decompose a tree then it has $O(\log n)$ diameter and any
LCL can be solved in $O(\log n)$ time on such a graph.
Thus, we mainly have to worry about the situation where \emph{compress} removes very
long ($\omega(1)$-length) paths.

The first observation is that it is easy to split up long degree-2 paths of level-$i$ vertices
into constant length paths, by artificially promoting a well-spaced subset of level-$i$ vertices to level $i+1$.
Thus, we have a situation that looks like this: level-$i$ vertices are arranged in an $O(1)$-length path,
each the root of a (colored) subtree of level-$(< i)$ vertices that were removed in previous rake/compress steps,
and bookended by level-$(>i)$ (black) vertices.
Call the subgraph between the bookends $H$.

\centerline{\scalebox{.4}{\includegraphics{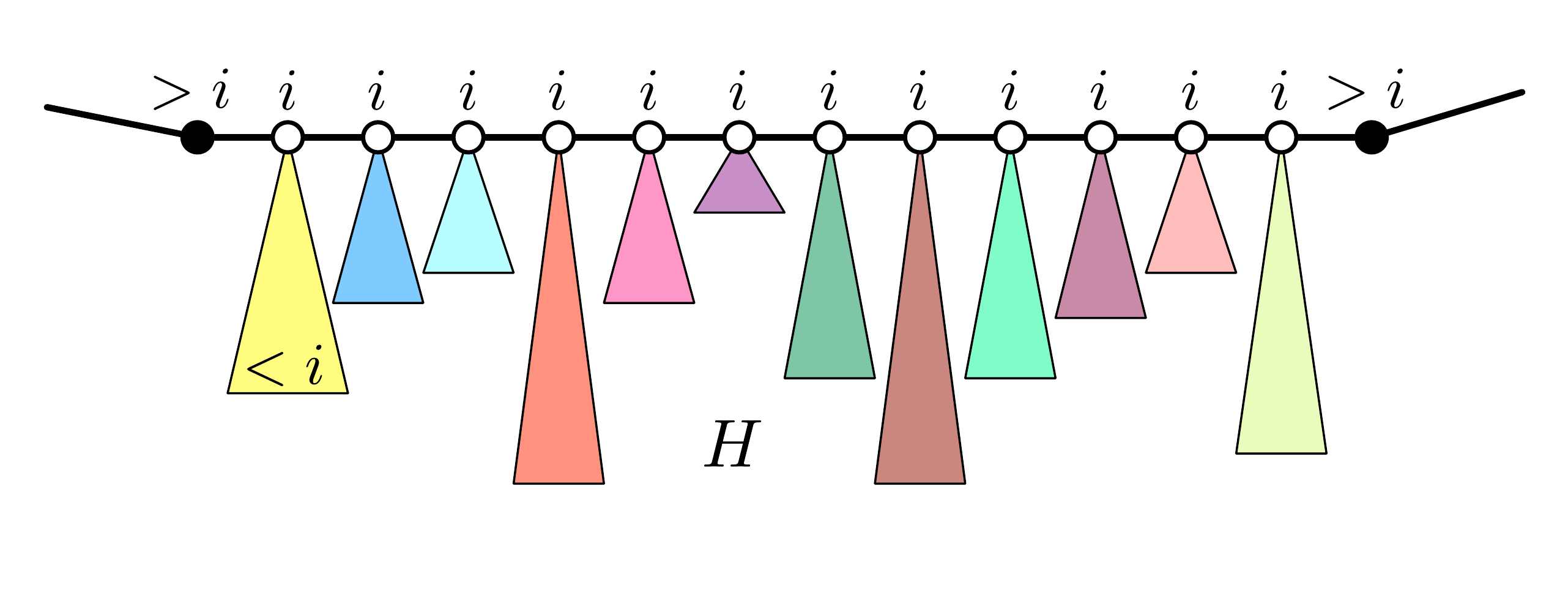}}}

In our approach it is the level-$(>i)$ vertices that are in charge of coordinating the labeling
of level-$(\le i)$ vertices in their purview.  In this diagram, $H$ is in the purview
of both black bookends.   We only have one tool available for computing a labeling of this subgraph:
an $n^{o(1)}$-time $\RandLOCAL$ algorithm $\mathcal{A}$ that works w.h.p.
What would happen if we \emph{simulated} $\mathcal{A}$ on $H$?
The simulation would fail catastrophically of course, 
since it needs to look up to an $n^{o(1)}$ radius, to parts of the graph far outside of $H$.

Note that the colored subtrees are unbounded in terms of size and depth.
Nonetheless, they fall into a \emph{constant} number of equivalence classes in the following sense.
The \emph{class} of a rooted tree is the set of all labelings of the $r$-neighborhood of its root
that can be extended to total labelings of the tree that are consistent with $\mathcal{P}$.

\centerline{\scalebox{.35}{\includegraphics{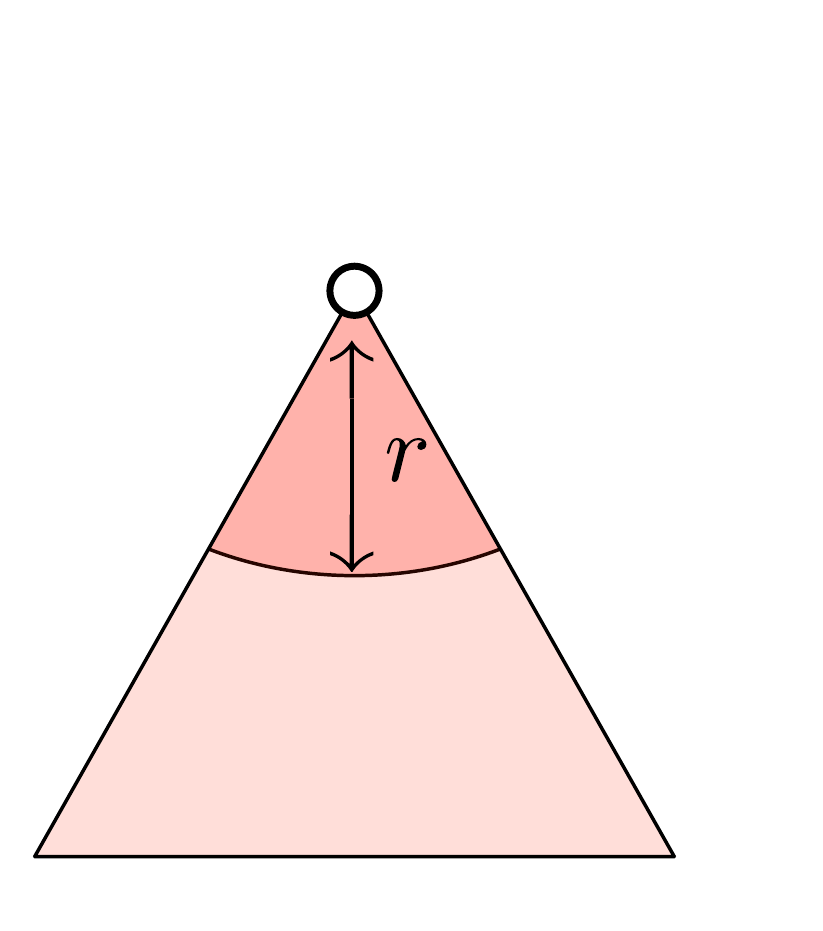}}}

In other words, the large and complex graph $H$ can be
succinctly encoded as a simple class vector $(c_1,c_2,\ldots,c_\ell)$, where $c_j$ is the class of the $j$th colored tree.
Consider the set of all labelings of $H$ that are consistent with $\mathcal{P}$.
This set can also be succinctly represented by listing the labelings of the $r$-neighborhoods of the bookends
that can be extended to all of $H$, while respecting $\mathcal{P}$.  The set of these partial labelings
defines the \emph{type} of $H$.
We show that $H$'s type
can be computed by a finite automaton that reads the class vector $(c_1,\ldots,c_\ell)$ one character at a time.
By the pigeonhole principle, if $\ell$ is sufficiently large then the automaton loops, meaning
that $(c_1,\ldots,c_\ell)$ can be written as $x\circ y\circ z$, which has the same type as every $x\circ y^j\circ z$,
for all $j\ge 1$.  This \emph{pumping lemma for trees} lets us dramatically expand the size of $H$
without affecting its type, i.e., how it interacts with the outside world beyond the bookends.

\centerline{\scalebox{.4}{\includegraphics{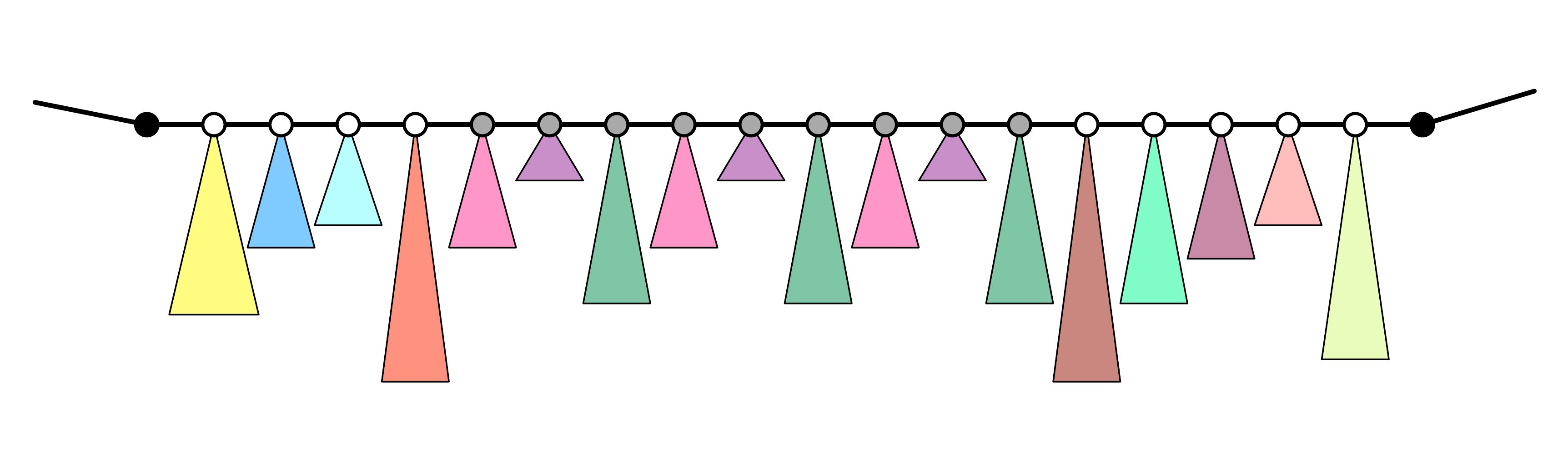}}}

This diagram illustrates the pumping lemma with a substring of 
$|y|=3$ trees (rooted at gray vertices) repeated $j=3$ times.
Now let us reconsider the simulation of $\mathcal{A}$.  
If we first pump $H$ to be long enough, 
and then simulate $\mathcal{A}$ on the middle section of pumped-$H$, 
$\mathcal{A}$ must, according to its $n^{o(1)}$ time bound,
compute a labeling \emph{without needing any information outside of pumped-$H$}, i.e., beyond the bookends.  
Thus, we can use
$\mathcal{A}$ to \emph{pre-commit} to a labeling of a small (radius-$r$) subgraph of 
pumped-$H$.  
Given this pre-commitment, the left and right bookends no longer need to coordinate their activities:
everything left (right) of the pre-committed zone is now in the purview of the left (right) bookend.
Interestingly, these manipulations (tree surgery and pre-commitments) can be repeated for each $i$,
yielding a hierarchy of \emph{imaginary} trees such that a proper 
labeling at one level of the hierarchy implies a proper labeling at the previous level.

\paragraph{Roadmap.}
This short proof sketch has been simplified to the point that it is riddled with small inaccuracies.
Nonetheless, it does accurately capture the difficulties, ideas, and techniques used in the actual proof.
In Section~\ref{sect:partial-labeling} we formally define the notion of a partially labeled graph,
i.e., one with certain vertices pre-commited to their output labels.
Section~\ref{sec.rel} defines a surgical ``cut-and-paste'' operation on graphs.
Section~\ref{sect:xi} defines a partition of the vertices of a subgraph $H$, 
which differentiates between vertices that ``see'' the outside graph, and those that see only $H$.
Section~\ref{sect:equiv} defines an equivalence relation on graphs that, intuitively,
justifies surgically replacing a subgraph with an equivalent graph.
Sections~\ref{sect:equiv-props} and~\ref{sect:equiv-classes}
explore properties of the equivalence relation.
Section~\ref{sect:aux} introduces the pumping lemma for trees,
and Section~\ref{sec.decomp} defines a specialized
Rake/Compress-style graph decomposition.
Section~\ref{sec.op} presents the operations $\extend$ (which pumps a subtree)
and $\labeling$ (which pre-commits a small partial labeling) in terms
of a black-box \emph{labeling function $f$}.
Section~\ref{sec.tree-construct} defines the set of all (partially labeled) trees that can be encountered,
by considering the interplay between the graph decomposition, $\extend$, and $\labeling$.
It is important that for each tree encountered, its partial labeling can be extended to a complete labeling consistent with $\mathcal{P}$;  whether this actually holds depends on the choice of black-box $f$.
Section~\ref{sect:feasible-functions}
shows how a \emph{feasible labeling function $f$} can be extracted from any $n^{o(1)}$-time algorithm 
$\mathcal{A}$ and Section~\ref{sec.tree} shows that $\mathcal{P}$ can be solved in $O(\log n)$ time,
given a feasible labeling function.

\subsection{Partial Labeled Graphs}\label{sect:partial-labeling}

A \emph{partially labeled graph} $\GG = (G,\LL)$ is a graph $G$ together with a function
$\LL : V(G) \rightarrow \LabelOut \cup \{\bottom\}$.  The vertices in $\LL^{-1}(\bottom)$ are \emph{unlabeled}.
A \emph{complete labeling} $\LL' : V(G)\rightarrow \LabelOut$ for $\GG$
is one that labels all vertices and is consistent with $\GG$'s partial labeling, i.e., $\LL'(v) = \LL(v)$ whenever $\LL(v)\neq\: \bottom$.
A \emph{legal} labeling is a complete labeling that is \emph{locally consistent} for all $v\in V(G)$,
i.e., the labeled subgraph induced by $N^r(v)$ is consistent with the LCL $\mathcal{P}$.
Here $N^r(v)$ is the set of all vertices within distance $r$ of $v$.

All graph operations can be extended naturally to partially labeled graphs.
For instance, a subgraph of a partially labeled graph $\GG=(G,\LL)$ is a pair $\HH=(H,\LL')$
such that $H$ is a subgraph of $G$, and $\LL'$ is $\LL$ restricted to the domain $V(H)$.
With slight abuse of notation, we usually write $\HH=(H,\LL)$.

\subsection{Graph Surgery\label{sec.rel}}

Let $\GG=(G,\LL)$ be a partially labeled graph, and let $\HH=(H,\LL)$ be a subgraph of $\GG$.
The \emph{poles} of $\HH$ are those vertices in $V(H)$ that are adjacent to some vertex in the
outside graph $V(G) - V(H)$.
We define an operation $\replace$ that surgically removes $\HH$ and replaces it with some $\HH'$.
\begin{description}
\item[$\replace$] Let $S=(v_1,\ldots,v_p)$ be a list of the poles of $\HH$
and let $S=(v_1',\ldots,v_p')$ be a designated set of poles in some partially labeled graph $\HH'$.
The partially labeled graph $\GG' = \replace(\GG,(\HH,S),(\HH',S'))$
is constructed as follows.
Beginning with $\GG$, replace $\HH$ with $\HH'$, and replace any edge $\{u,v_i\}$, $u\in V(G) -  V(H)$,
with $\{u,v_i'\}$.  If the poles $S,S'$ are clear from context,
we may also simply write $\GG' = \replace(\GG,\HH,\HH')$.
Writing $\GG'=(G',\LL')$ and $\HH' =(H',\LL')$, there is a natural 1-1 correspondence between the
vertices in $V(G) - V(H)$ and $V(G') - V(H')$.
\end{description}

In the proof of our $n^{o(1)}\rightarrow O(\log n)$ speedup thereom we only consider unipolar and bipolar graphs
($p \in \{1,2\}$) but for maximum generality we define everything w.r.t.~graphs having $p\ge 1$ poles.\\

\centerline{\scalebox{.35}{\includegraphics{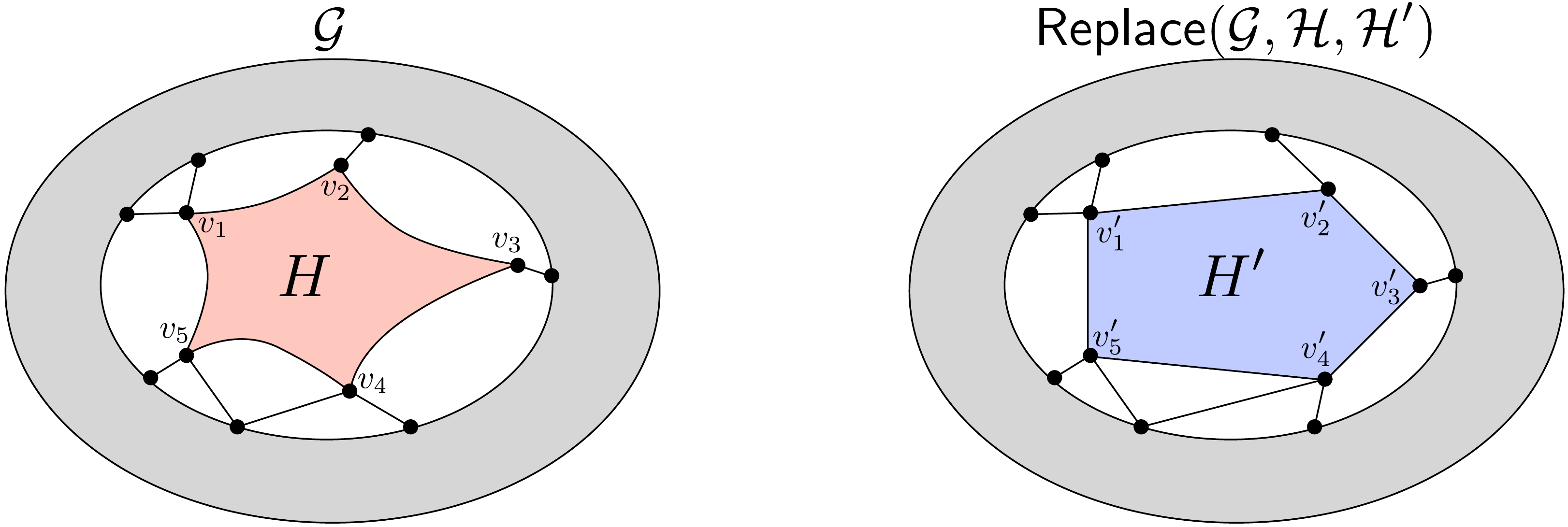}}}
\vspace{.3cm}

Given a legal labeling $\LL_\diamond$ of $\GG$, we would like to know whether there is a legal labeling $\LL_\diamond'$ of $\GG'$
that agrees with $\LL_\diamond$, i.e., $\LL_\diamond(v) = \LL_\diamond'(v')$ for each $v\in V(G) - V(H)$ and the corresponding
$v' \in V(G') - V(H')$.  Our goal is to define an equivalence relation $\simm$ on partially labeled graphs (with designated poles)
so that the following is true:
if $(\HH,S) \simm (\HH',S')$, then such a legal labeling $\LL_\diamond'$ must exist, regardless of the choice of $\GG$ and $\LL_\diamond$.
Observe that since $\mathcal{P}$ has radius $r$, the interface between $V(H)$ (or $V(H')$) and the rest of the graph only
occurs around the $O(r)$-neighborhoods of the poles of $\HH$ (or $\HH'$).
This motivates us to define a certain partition of $\HH$'s vertices that depends on its poles and $r$.

\subsection{A Tripartition of the Vertices}\label{sect:xi}

Let $\HH=(H,\LL)$ be a partially labeled graph with poles $S = (v_1, \ldots, v_p)$.
Define $\xi(\HH,S)=(D_1,D_2,D_3)$ to be a tripartition of $V(H)$, where
$D_1 = \bigcup_{v \in S} N^{r-1}(v)$,
$D_2 = \bigcup_{v \in D_1} N^r(v) - D_1$,
and $D_3 = V(H) - (D_1 \cup D_2)$.
See Figure~\ref{fig:D-sets} for an illustration.

\begin{figure}[h]
\centerline{\scalebox{.5}{\includegraphics{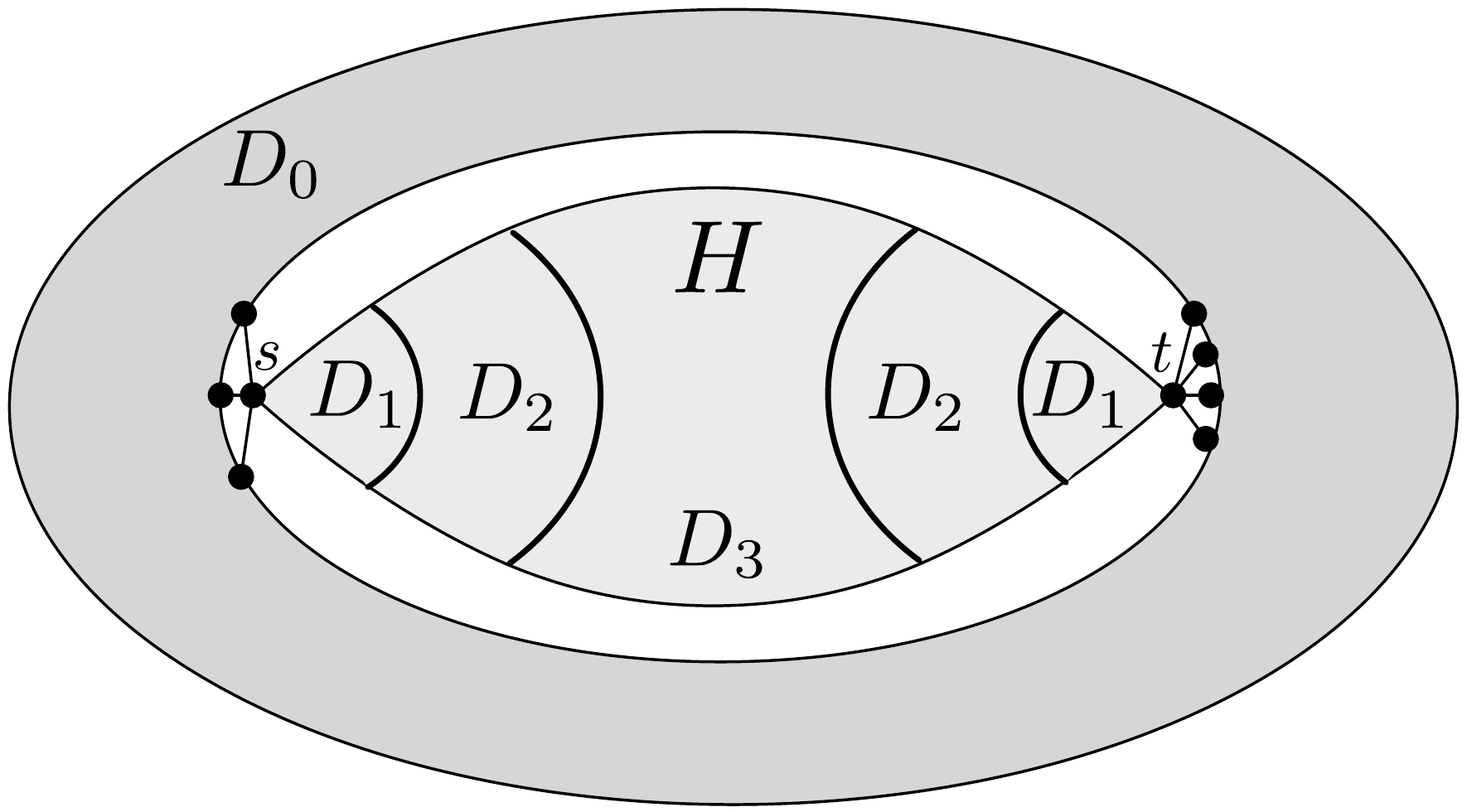}}}
\caption{\label{fig:D-sets}A partially labeled subgraph $\HH$ with poles $S=(s,t)$,
embedded in a larger graph $\GG$.
In the partition $\xi(\HH,S) = (D_1,D_2,D_3)$, $D_1$ is the set of vertices in $V(H)$ within radius $r-1$
of $S$, $D_2$ are those within radius $2r-1$ of $S$, excluding $D_1$, and $D_3$ is the rest of $V(H)$.
When $\HH$ is embedded in some larger graph $\GG$,
$D_0$ denotes the remaining vertices in $V(G) - V(H)$.}
\end{figure}

Consider the partition $\xi(\HH,S)=(D_1,D_2,D_3)$ of a partially labeled graph $\HH = (H,\LL)$.
Let $\LL_\ast : D_1 \cup D_2 \rightarrow \LabelOut$ assign output labels to $D_1\cup D_2$.
We say that $\LL_\ast$ is {\em extendible} (to all of $V(H)$) if there exists a complete labeling
$\LL_\diamond$ of $H$ such that $\LL_\diamond$ agrees with $\LL$ where it is defined,
agrees with $\LL_\ast$ on $D_1\cup D_2$, and is locally consistent with $\mathcal{P}$ on all vertices
in $D_2\cup D_3$.\footnote{We are not concerned whether $\LL_\diamond$ is consistent with $\mathcal{P}$
for vertices in $D_1$.  Ultimately, $\HH$ will be a subgraph of a larger graph $\GG$.  Since
the $r$-neighborhoods of vertices in $D_1$ will intersect $V(G) - V(H)$, the labeling of $H$
does not provide enough information to tell if these vertices' $r$-neighborhoods will be consistent with $\mathcal{P}$.  See Figure~\ref{fig:D-sets}.}

\subsection{An Equivalence Relation on Graphs}\label{sect:equiv}

Consider two partially labeled graphs $\HH$ and $\HH'$
with poles $S=(v_1, \ldots, v_p)$ and $S'=(v_1', \ldots, v_p')$, respectively.
Let $\xi(\HH,S)=(D_1,D_2,D_3)$ and $\xi(\HH',S')=(D_1',D_2',D_3')$.
Define $\QQ=(Q,\LL)$ and $\QQ'=(Q',\LL')$ as the subgraphs of $\HH$ and $\HH'$
induced by the vertices in $D_1 \cup D_2$ and $D_1'\cup D_2'$, respectively.

The relation $(\HH,S) \simm (\HH',S')$ holds if and only if there is a 1-1 correspondence
$\phi : (D_1 \cup D_2)\rightarrow (D_1'\cup D_2')$ meeting the following conditions.

\begin{description}
\item[Isomorphism.]
The two graphs $Q$  and $Q'$ are isomorphic under $\phi$.
Moreover,
for each $v \in D_1 \cup D_2$ and its corresponding vertex $v' = \phi(v) \in D_1' \cup D_2'$,
(i) $\LL(v) = \LL'(v')$,
(ii) if the underlying LCL problem has input labels, then the input labels of $v$ and $v'$ are the same,
and
(iii) $v$ is the $i$th pole in $S$ iff $v'$ is the $i$th pole in $S'$.
\item[Extendibility.]
Let $\LL_\ast$ be {\em any} assignment of output labels to vertices in $D_1 \cup D_2$
and let $\LL_\ast'$ be the corresponding labeling of $D_1' \cup D_2'$ under $\phi$.
Then $\LL_\ast$ is extendible to $V(H)$ if and only if $\LL_\ast'$ is extendible to $V(H')$.
\end{description}
Notice that there could be many 1-1 correspondences between $D_1 \cup D_2$ and $D_1' \cup D_2'$ that satisfy the isomorphism
requirement, though only some subset may satisfy the
extendibility requirement due to differences in the topology and partial labeling of $D_3$ and $D_3'$.
Any $\phi$ meeting both requirements is a {\em witness} of the relation $(\HH,S) \simm (\HH',S')$.

\subsection{Properties of the Equivalence Relation}\label{sect:equiv-props}

Let $\GG' = \replace(\GG,(\HH,S),(\HH',S'))$.
Consider the partitions
$\xi(\HH,S)=(D_1,D_2,D_3)$ and $\xi(\HH',S')=(D_1',D_2',D_3')$
and let $D_0 = V(G) - V(H)$ and $D_0' = V(G') - V(H')$
be the remaining vertices in $G$ and $G'$, respectively.

If $(\HH,S) \simm (\HH',S')$
then there exists a 1-1 correspondence 
$\phi \;:\; (D_0 \cup D_1 \cup D_2) \rightarrow (D_0' \cup D_1' \cup D_2')$
such that
(i) $\phi$ restricted to $D_0$ is the natural 1-1 correspondence between $D_0$ and $D_0'$
and
(ii) $\phi$ restricted to $D_1 \cup D_2$ witnesses the relation $(\HH,S) \simm (\HH',S')$.
Such a 1-1 correspondence $\phi$ is called {\em good}.
We have the following lemma.

\begin{lemma}\label{lem:base}
Let $\GG' = \replace(\GG,(\HH,S),(\HH',S'))$.
Consider the partitions
$\xi(\HH,S)=(D_1,D_2,D_3)$ and $\xi(\HH',S')=(D_1',D_2',D_3')$
and let $D_0 = V(G) - V(H)$ and $D_0' = V(G') - V(H')$.
Suppose that $(\HH,S) \simm (\HH',S')$, so there is a \emph{good} 1-1 correspondence
$\phi : (D_0 \cup D_1 \cup D_2)\rightarrow (D_0' \cup D_1' \cup D_2')$.
Let $\LL_\diamond$ be a complete labeling of $\GG$ that is locally consistent for all vertices in $D_2 \cup D_3$.
Then there exists a complete labeling $\LL_\diamond'$ of $\GG'$ such that the following conditions are met.
\begin{description}
\item[Condition 1.] $\LL_\diamond(v) = \LL_\diamond'(v')$ for each $v \in D_0 \cup D_1 \cup D_2$ and its corresponding vertex $v'=\phi(v) \in D_0' \cup D_1' \cup D_2'$.  Moreover, if $\LL_\diamond$ is locally consistent for $v$, then $\LL_\diamond'$ is locally consistent for $v'$.
\item[Condition 2.] $\LL_\diamond'$ is locally consistent for all vertices in $D_2' \cup D_3'$.
\end{description}
\end{lemma}

\begin{proof}
We construct $\LL_\diamond'$  as follows.
First of all, for each $v \in D_0 \cup D_1 \cup D_2$, fix
$\LL_\diamond'(\phi(v)) = \LL_\diamond(v)$.
It remains to show how to assign output labels to vertices in $D_3'$  to meet Conditions 1 and 2.

Let $\LL_\ast$ be $\LL_\diamond$ restricted to the domain $D_1 \cup D_2$.
Similarly, let $\LL_\ast'$ be $\LL_\diamond'$ restricted to $D_1' \cup D_2'$.
Due to the fact that $\LL_\diamond$ is locally consistent for all vertices in $D_2 \cup D_3$, the labeling $\LL_\ast$ is extendible to all of $\HH$.
Since $(\HH,S) \simm (\HH',S')$, the labeling $\LL_\ast'$ must also be extendible to all of $\HH'$.
Thus, we can set $\LL_\diamond'(v')$ for all $v' \in D_3'$ in such a way that $\LL_\diamond'$ is locally consistent for all vertices in $D_2' \cup D_3'$. Therefore, Condition 2 is met.

To see that (the second part of) Condition 1 is also met,
observe that for $v\in D_0\cup D_1$, $N^r(v) \subseteq D_0 \cup D_1 \cup D_2$.
Therefore, if $\LL_\diamond$ is locally consistent for $v\in D_0 \cup D_1$,
then $\LL_\diamond'$ is locally consistent for $\phi(v)$ since they have the same
radius-$r$ neighborhood view.
Condition 2 already guarantees that $\LL_\diamond'$ is locally consistent for all 
$v' \in D_2'$.\footnote{It is this lemma that motivates our definition of the \emph{tri}partition $\xi(\HH,S)$.  It is not clear how an
analogue of Lemma~\ref{lem:base} could be proved using the seemingly more natural \emph{bi}partition,
i.e., by collapsing $D_1,D_2$ into one set.}
\end{proof}

Theorem~\ref{thm:rel-1} provides a user-friendly corollary of Lemma~\ref{lem:base},
which does not mention the tripartition $\xi$.

\begin{theorem}\label{thm:rel-1}
Let $\GG = (G,\LL)$ and $\HH=(H,\LL)$ be a subgraph $\GG$.
Suppose $\HH'$ is a graph for which $(\HH,S)\simm (\HH',S')$
and let $\GG' = \replace(\GG,(\HH,S),(\HH',S'))$.
We write $\GG' = (G',\LL')$ and $\HH'=(H',\LL')$.
Let $\LL_\diamond$ be a complete labeling of $\GG$ that is locally consistent for all vertices in $H$.
Then there exists a complete labeling $\LL_\diamond'$ of $\GG'$ such that the following conditions are met.
\begin{itemize}
\item  For each $v \in V(G) - V(H)$ and its corresponding $v' \in V(G') - V(H')$, we have $\LL_\diamond(v) = \LL_\diamond'(v')$. Moreover, if $\LL_\diamond$ is locally consistent for $v$, then $\LL_\diamond'$ is locally consistent for $v'$.
\item  $\LL_\diamond'$ is locally consistent for all vertices in $H'$.
\end{itemize}
\end{theorem}

Theorem~\ref{thm:rel-1} has several useful consequences.
If $\LL_\diamond$  is a legal labeling of $\GG$, then the output labeling $\LL_\diamond'$ of $\GG'$ guaranteed by Theorem~\ref{thm:rel-1} is also legal.
Observe that setting $\GG=\HH$ in Theorem~\ref{thm:rel-1} implies $\GG'=\HH'$.
Suppose that $\HH$ admits a legal labeling.
For {\em any} $(\HH',S')$ such that $(\HH',S') \simm (\HH,S)$,
the partially labeled graph $\HH'$ also admits a legal labeling.
Thus, whether $\HH$ admits a legal labeling is determined by the equivalence class of $(\HH,S)$ (for any choice of $S$).

Roughly speaking, Theorem~\ref{thm:rel-2} shows that the equivalence class of $(\GG,X)$ is preserved
after replacing a subgraph $\HH$ of $\GG$ by another partially labeled graph $\HH'$
such that $(\HH,S) \simm (\HH',S')$.

\begin{theorem}\label{thm:rel-2}
Let $\GG = (G,\LL)$ and let $\HH=(H,\LL)$ be a subgraph $\GG$.
Suppose $\HH'$ is such that $(\HH,S) \simm (\HH',S')$ for some pole lists $S,S'$.
Let $\GG' = \replace(\GG,(\HH,S),(\HH',S'))$ be a partially labeled graph.
Designate a set $X \subseteq (V(G) - V(H)) \cup S$ as the poles of $\GG$, listed in some order,
and let $X'$ be the corresponding list of vertices in $\GG'$.
It follows that $(\GG,X) \simm (\GG',X')$.
\end{theorem}
\begin{proof}
Consider the four partitions $\xi(\HH,S)=(B_1,B_2,B_3)$,  $\xi(\HH',S')=(B_1',B_2',B_3')$,
$\xi(\GG,X)=(D_1,D_2,D_3)$, and  $\xi(\GG',X')=(D_1',D_2',D_3')$.
We write $B_0 = V(G) - V(H)$ and $B_0' =  V(G') - V(H')$.
Let $\phi$ be any good 1-1 correspondence from  $B_0 \cup B_1 \cup B_2$ to $B_0' \cup B_1' \cup B_2'$.
Because $X \subseteq B_0 \cup S$, we have $D_1 \cup D_2 \subseteq B_0 \cup B_1 \cup B_2$ and $D_1' \cup D_2' \subseteq B_0' \cup B_1' \cup B_2'$.
To show that $(\GG,X) \simm (\GG',X')$, it suffices to prove that $\phi$ (restricted to the domain $D_1 \cup D_2$)
is a witness to the relation $(\GG,X) \simm (\GG',X')$.

Let $\LL_\ast : (D_1 \cup D_2)\rightarrow \LabelOut$ and $\LL_\ast'$ 
be the corresponding labeling of $D_1'\cup D_2'$.
All we need to do is show that  $\LL_\ast$ is extendible to all of $V(G)$ if and only if $\LL_\ast'$ is extendible to all of $V(G')$.
Since we can also write $\GG = \replace(\GG',(\HH',S'),(\HH,S))$,
it suffices to show just one direction, i.e., if $\LL_\ast$ is extendible then $\LL_\ast'$ is extendible.

Suppose that $\LL_\ast$ is extendible. Then there exists an output labeling $\LL_\diamond$ of $\GG$ such that (i) for each $v \in D_1 \cup D_2$, we have $\LL_\ast(v) = \LL_\diamond(v)$, and (ii) $\LL_\diamond$  is locally consistent for all vertices in $D_2 \cup D_3$.
Observe that  $D_2 \cup D_3 \supseteq B_2 \cup B_3$.
By Lemma~\ref{lem:base}, there exists a complete labeling $\LL_\diamond'$ of $\GG'$ such that the two conditions in Lemma~\ref{lem:base} are met. We show that this implies that $\LL_\ast'$ is extendible.

Lemma~\ref{lem:base} guarantees that $\LL_\diamond(v) = \LL_\diamond'(\phi(v))$ for each $v \in B_0 \cup B_1 \cup B_2$ and its corresponding vertex $\phi(v) \in B_0' \cup B_1' \cup B_2'$. 
Since $D_1' \cup D_2' \subseteq B_0' \cup B_1' \cup B_2'$, we have $\LL_\ast'(v') = \LL_\diamond'(v')$ 
for each $v' \in D_1' \cup D_2'$.

Since $\LL_\diamond$  is locally consistent for all vertices in $D_2 \cup D_3$, Lemma~\ref{lem:base} guarantees that $\LL_\diamond'$ is locally consistent for all vertices in $D_2' \cup D_3'$. More precisely, due to Condition 1, $\LL_\diamond'$ is locally consistent for all vertices in $(D_2' \cup D_3') - B_3'$; due to Condition 2, $\LL_\diamond'$ is locally consistent for all vertices in $B_2' \cup B_3'$.

Thus, $\LL_\ast'$ is extendible, as the complete labeling $\LL_\diamond'$ of $\GG'$ satisfies: (i)  for each $v' \in D_1' \cup D_2'$, we have $\LL_\ast'(v') = \LL_\diamond'(v')$, and (ii) $\LL_\diamond'$  is locally consistent for all vertices in $D_2' \cup D_3'$.
\end{proof}

\subsection{The Number of Equivalence Classes}\label{sect:equiv-classes}

An important feature of $\simm$ is that it has a \emph{constant} number of equivalence classes,
for any fixed number $p$ of poles.
Which constant is not important, but we shall work out an upper bound nonetheless.\footnote{For the sake of simplicity, in the calculation we assume that the underlying LCL problem does not refer to port-numbering. It is straightforward to see that even if port-numbering is taken into consideration, the number of equivalence classes (for any fixed $p$) is still a constant.}

Consider a partially labeled graph $\HH$ with poles $S=(v_1, \ldots, v_p)$.
Let $\xi(\HH,S)=(D_1,D_2,D_3)$ and define $\QQ=(Q,\LL)$ to be the subgraph of $\HH$ induced by $D_1 \cup D_2$.
Observe that the equivalence class of $(\HH,S)$ is determined by
(i) the topology of $Q$ (including its input labels from $\LabelIn$, if $\mathcal{P}$ has input labels),
(ii) the locations of the poles $S \subseteq V(Q)$ in $Q$,
and
(iii) the subset of all output labelings of $V(Q) = D_1 \cup D_2$ that are extendible.

The number of vertices in $D_1 \cup D_2$ is at most $p \Delta^{2r}$.
The total number of distinct graphs of at most $p \Delta^{2r}$ vertices (with input labels from $\LabelIn$ and a set of $p$ designated poles)
is at most $2^{p \Delta^{2r} \choose 2} \left|\LabelIn \right|^{p \Delta^{2r}}$.
The total number of output labelings of $D_1 \cup D_2$ is at most $\left|\LabelOut \right|^{p \Delta^{2r}}$.
Therefore, the total number of equivalence classes of graphs with $p$ poles is at most
$2^{p \Delta^{2r} \choose 2} \left|\LabelIn \right|^{p \Delta^{2r}} 2^{\left|\LabelOut \right|^{p \Delta^{2r}}}$,
which is constant whenever $\Delta, r, \left|\LabelIn \right|, \left|\LabelOut \right|$, and $p$ are.

\newcommand{\classSet}{\mathscr{C}}
\newcommand{\RR}{\mathcal{R}}

\subsection{A Pumping Lemma for Trees}\label{sect:aux}

In this section we consider partially labeled trees with one and two poles; they are called \emph{unipolar} (or \emph{rooted}) and \emph{bipolar}, respectively.
Let $\TT=(T,\LL)$ be a unipolar tree with pole list $S=(z)$, $z \in V(T)$.
Define $\class(\TT)$ to be the equivalence class of $(\TT, S)$ w.r.t.~$\simm$.
Notice that whether a partially labeled rooted tree $\TT$ admits a legal labeling is determined by $\class(\TT)$ (Theorem~\ref{thm:rel-1}).
We say that a class is {\em good} if each partially labeled rooted tree in the class admits a legal labeling;
otherwise the class is {\em bad}.
We write $\classSet$ to denote the set of all classes. Notice that $|\classSet|$ is constant.
The following lemma is a specialization of Theorem~\ref{thm:rel-2}.

\begin{lemma}\label{lem:replace-rootedtree}
Let $\TT$ be a partially labeled rooted (unipolar) tree, and let $\TT'$ be a rooted subtree of $\TT$,
whose leaves are also leaves of $\TT$.
Let $\TT''$ be another partially labeled rooted tree such that $\class(\TT') = \class(\TT'')$.
Then replacing $\TT'$ with $\TT''$ does not alter the class of $\TT$.
\end{lemma}

Let $\HH=(H,\LL)$ be a bipolar tree with poles $S=(s,t)$.
The unique oriented path in $H$ from $s$ to $t$ is called the {\em core path} of $\HH$.
It is more convenient to express a bipolar tree as a \emph{sequence of rooted/unipolar trees}, as follows.
The partially labeled bipolar tree $\HH = (\TT_i)_{i \in [k]}$
is formed by arranging the roots of unipolar trees $(\TT_i)$ into a path $P=(v_1,\ldots,v_k)$, where $v_i$ is the root/pole of $\TT_i$.
The two poles of $\HH$ are $s=v_1$ and $t = v_k$, so $P$ is the core path of $\HH$.
Define $\type(\HH)$ as the equivalence class of $(\HH,S=(s,t))$ w.r.t. $\simm$.
The following lemma follows from Theorem~\ref{thm:rel-2}.

\begin{lemma}\label{lem:type2class}
Let $\HH$ be a partially labeled bipolar tree with poles $(s,t)$. Let $\TT$ be $\HH$, but regarded as a unipolar tree rooted at $s$.
Then $\class(\TT)$ is determined by $\type(\HH)$.
If we write $\HH = (\TT_i)_{i \in [k]}$, then $\type(\HH)$ is determined by $\class(\TT_1), \ldots, \class(\TT_{k})$.
\end{lemma}

Let $\GG=(G,\LL)$ be a partially labeled graph, and let $\HH=(H, \LL)$ be a
bipolar subtree of $\GG$ with poles $(s,t)$.
Let $\HH'$ be another partially labeled bipolar tree.
Recall that $\GG' = \replace(\GG, \HH, \HH')$ is defined as the partially labeled graph resulting from replacing the subgraph $\HH$ with $\HH'$ in $\GG$.
We write $\GG' = (G', \LL')$ and  $\HH'=(H', \LL')$.
The following lemmas follow from Theorems~\ref{thm:rel-1} and~\ref{thm:rel-2}.

\begin{lemma}\label{lem:replace}
Consider $\GG' = \replace(\GG,\HH,\HH')$.
If $\type(\HH') = \type(\HH)$ and $\GG$ admits a legal labeling $\LL_\diamond$,
then $\GG'$ admits a legal labeling $\LL_\diamond'$ such that
${\LL_\diamond}(v) = \LL_\diamond'(v')$ for each vertex $v \in V(G) - V(H)$ and its corresponding $v' \in V(G') - V(H')$.
\end{lemma}

\begin{lemma}\label{lem:replace-2}
Suppose that $\GG = (\TT_i)_{i \in [k]}$ is a partially labeled bipolar tree,
$\HH = (\TT_i, \ldots, \TT_j)$ is a bipolar subtree of $\GG$,
and $\HH'$ is some other partially labeled bipolar tree with $\type(\HH') = \type(\HH)$.
Then $\GG' = \replace(\GG,\HH,\HH')$ is a partially labeled bipolar tree
and $\type(\GG') = \type(\GG)$.
\end{lemma}

\begin{lemma} \label{lem:type}
Let $\HH = (\TT_i)_{i \in [k]}$ and
$\HH' = (\TT_i)_{i \in [k+1]}$ be identical to $\HH$ in its first $k$ trees.
Then $\type(\HH')$ is a function of $\type(\HH)$ and $\class(\TT_{k+1})$.
\end{lemma}

Lemma~\ref{lem:type} is what allows us to bring classical automata theory into play.
Suppose that we somehow computed and stored $c_i = \class(\TT_i)$ at the root of $\TT_i$.
Lemma~\ref{lem:type} implies that a finite automaton walking along the core path of $\HH' = (\TT_i)_{i\in [k+1]}$,
can compute $\type(\HH')$, by reading the vector $(c_1,\ldots,c_{k+1})$ one character at a time.
The number of states in the finite automaton depends only on the number of types (which is constant)
and is independent of $k+1$ and the size of the individual trees $(\TT_i)$.
Define $\Lpump=O(1)$ as the number of states in this finite automaton.
The following {\em pumping lemma} for bipolar trees is analogous to the pumping lemma
for regular languages.

\begin{lemma}\label{thm:pump}
Let $\HH = (\TT_1, \ldots, \TT_k)$, with $k \geq \Lpump$.
We regard each $\TT_i$ in the string notation $\HH = (\TT_1, \ldots, \TT_k)$ as a character.
 Then $\HH$ can be decomposed into three substrings $\HH = x \circ y \circ z$ such that (i) $|xy| \leq \Lpump$, (ii) $|y|\geq 1$, and (iii) $\type(x \circ y^j \circ z) = \type(\HH)$ for each non-negative integer $j$.
\end{lemma}

We will use Lemma~\ref{thm:pump} to expand the length of the core path of a bipolar tree to be close to a
desired \emph{target length} $w$.  The specification for the function $\pump$ is as follows.
\begin{description}
\item[$\pump$]
Let $\HH = (\TT_i)_{i \in [k]}$ be a partially labeled bipolar tree with $k \geq \Lpump$.
$\pump(\HH,w)$ produces a partially labeled bipolar tree $\HH' = (\TT_i')_{i \in [k']}$
such that
(i) $\type(\HH) = \type(\HH')$,
(ii) $k' \in [w,w+\Lpump]$,
and
(iii) if we let $Z=\{\TT_i\}_{i\in [k]}$ (resp., $Z' = \{\TT_i'\}_{i\in [k']}$) be the \underline{\emph{set}}
of rooted trees appearing in the tree list of
$\HH$ (resp., $\HH'$), then $Z' = Z$.
\end{description}
By Lemma~\ref{thm:pump}, such a function $\pump$ exists.

\subsection{Rake \& Compress Graph Decomposition\label{sec.decomp}}

In this section we describe an $O(\log n)$-round $\DetLOCAL$ algorithm to
decompose the vertex set $V(G)$ of a tree into the disjoint union $V_1\cup\cdots \cup V_L$,  $L=O(\log n)$.
Our algorithm is inspired by Miller and Reif's {\em parallel tree contraction}~\cite{MillerR89}.
We first describe the decomposition algorithm then analyze its properties.

Fix the constant $\ell = 2(r + \Lpump)$, where $r,\Lpump$
depend on the LCL problem $\mathcal{P}$.
In the \emph{postprocessing} step of the decomposition algorithm we compute an
\emph{$(\ell,2\ell)$-independent set}, in $O(\log^* n)$ time~\cite{Linial92}, defined as follows.

\begin{definition}
Let $P$ be a path.  A set $I\subset V(P)$ is called an \emph{$(\alpha,\beta)$-independent set}
if the following conditions are met:
(i) $I$ is an independent set, and $I$ does not contain either endpoint of $P$, and
(ii) each connected component induced by $V(P)- I$ has at least $\alpha$ vertices and at most $\beta$ vertices,
unless $|V(P)|<\alpha$, in which case $I=\emptyset$.
\end{definition}

\paragraph{The Decomposition Algorithm.}
The algorithm begins with $U=V(G)$ and $i=1$,
repeats Steps 1--3 until $U=\emptyset$, then executes the \emph{Postprocessing} step.
\begin{enumerate}
\item For each $v \in U$:
 \begin{enumerate}
  \item {\sf Compress.} If $v$ belongs to a path $P$ such that $|V(P)|\geq \ell$ and $\deg_U(u) = 2$ for each $u \in V(P)$, then tag $v$ with $i_C$.
 \item {\sf Rake.} If $\deg_U(v)  = 0$, then tag $v$ with $i_R$. If $\deg_U(v)  = 1$ and the unique neighbor $u$ of $v$ in $U$ satisfies either (i) $\deg_U(u) > 1$ or (ii) $\deg_U(u) =1 $ and $\ID(v)  > \ID(u)$, then tag $v$ with $i_R$.
 \end{enumerate}
\item Remove from $U$ all vertices tagged $i_C$ or $i_R$.
\item $i \leftarrow i+1$.
\end{enumerate}

\paragraph{Postprocessing Step.}
Initialize $V_i$ as the set of all vertices tagged $i_C$ or $i_R$.  At this point the graph induced by
$V_i$ consists of \emph{unbounded length} paths, but we prefer constant length paths.
For each edge $\{u,v\}$ such that $v$ is tagged $i_R$ and $u$ is tagged $i_C$,
promote $v$ from $V_i$ to $V_{i+1}$.
For each path $P$ that is a connected component induced by vertices tagged
$i_C$, compute an $(\ell,2\ell)$-independent set $I_P$ of $P$,
and then promote every vertex in $I_P$
from $V_i$ to $V_{i+1}$.\footnote{The set $V_i$ in the graph decomposition is analogous to (but clearly different from) the set $V_i$ defined in the Hierarchical $2\f{1}{2}$-coloring problem from Section~\ref{sec.poly}.}

\paragraph{Properties of the Decomposition.}
As we show below, $L=O(\log n)$ iterations suffice, i.e.,
$V(G) = V_1\cup \cdots\cup V_L$.
The following properties are easily verified.
 \begin{itemize}
 \item Define $\Gd{i}$ as the graph induced by vertices at level $i$ or above: $\bigcup_{j=i}^L V_j$.
For each $v\in V_i$, $\deg_{\Gd{i}}(v) \leq 2$.

 \item Define $\Pset_i$ as the set of connected components (paths) induced by vertices in $V_i$ that contain more than one vertex.
 For each $P \in \Pset_i$,  $\ell \leq |V(P)| \leq 2\ell$
 and $\deg_{\Gd{i}}(v) = 2$ for each vertex $v \in V(P)$.

 \item The graph $\Gd{L}$ contains only isolated vertices, i.e., $\Pset_L = \emptyset$.
 \end{itemize}

As a consequence, each vertex $v \in V_i$ falls into exactly one of two cases:
(i) $v$ has $\deg_{\Gd{i}}(v) \leq 1$ and has no neighbor in $V_i$,
or
(ii) $v$ has $\deg_{\Gd{i}}(v) = 2$ and is in some path $P \in \Pset_i$.

\paragraph{Analysis.} We prove that for $L = O(\log_{1+1/\ell} n) = O(\log n)$,
$L$ iterations of the graph decomposition routine suffices to decompose any $n$-vertex tree.
Each iteration of the routine takes $O(1)$ time, and the $(\ell,2\ell)$-independent set
computation at the end takes $O(\log^\ast n)$ time, so $O(\log n)$ time suffices in $\DetLOCAL$.

Let $W$ be the vertices of a connected component induced by $U$ at the beginning of the $i$th iteration.\footnote{In general, the graph induced by $U$ is a forest. It is simpler to analyze a single connected component $W$.}
We claim that at least a constant $\Omega(1/\ell)$ fraction of vertices in
$W$ are eliminated (i.e., tagged $i_C$ or $i_R$) in the $i$th iteration.
The proof of the claim is easy for the special case of $\ell = 1$, as follows. If $W$ is not a single edge, then all
$v \in W$ with $\deg_U(v) \leq 2$ are eliminated. Since the degree of at least half of the vertices in a tree is at most 2, the claim follows.   In general, degree-2 paths of length less than $\ell$ are not eliminated quickly.  If one endpoint of such a path is a leaf, vertices in the path are peeled off by successive \emph{Rake} steps.

Assume w.l.o.g.~that $|W| > 2 (\ell + 1)$.
Define $W_1 = \{v \in W \;|\; \deg_U(v) = 1\}$, $W_2 = \{v \in W \;|\; \deg_U(v) = 2\}$, and $W_3= \{v \in W \;|\; \deg_U(v) \geq 3\}$.
\begin{description}
\item[Case 1:] $|W_2| \geq \frac{\ell  |W|}{\ell + 1}$. The number of connected components induced by vertices in $W_2$ is at most $|W_1|+|W_3|-1 < \frac{|W|}{\ell + 1}$. The number of vertices in $W_2$ that are not tagged $i_C$ during {\sf Compress} is less than $\frac{(\ell - 1)|W|}{\ell + 1}$. Therefore, at least $\frac{\ell  |W|}{\ell + 1} - \frac{(\ell - 1)|W|}{\ell + 1} = \frac{|W|}{\ell + 1}$ vertices are tagged $i_C$ by {\sf Compress}.
\item[Case 2:] $|W_2| < \frac{\ell  |W|}{\ell + 1}$. In any tree $|W_1| > |W_3|$, so
$|W_1| > \frac{|W_1|+|W_3|}{2} = \frac{|W|-|W_2|}{2} \geq \frac{|W|}{2(\ell + 1)}$.
Therefore, at least $\frac{|W|}{2(\ell + 1)}$ vertices are tagged $i_R$ by {\sf Rake}.
\end{description}
Hence the claim follows.

\subsection{$\extend$ and $\labeling$ Operations \label{sec.op}}

In this section we define two operations $\extend$ and $\labeling$ which are used extensively
in Sections~\ref{sec.tree-construct}---\ref{sec.tree}.
The operation $\extend$ is parameterized by a target length $w \geq \ell = 2(r+\Lpump)$.
The operation $\labeling$ is parameterized by a function $f$ which takes a partially labeled bipolar tree
$\HH$ as input, and assigns output labels to the vertices in $v \in N^{r-1}(e)$,
where $e$ is the middle edge in the core path of $\HH$.\footnote{By definition,
if $e=\{x,y\}$ then $N^{r-1}(e) = N^{r-1}(x)\cup N^{r-1}(y)$.}
\begin{description}
\item[\labeling.]
Let $\HH=(\TT_1, \ldots, \TT_x)$ be a partially labeled bipolar tree with $x \geq \ell$.
Let $(v_1, \ldots, v_x)$ be the core path of $\HH$
and $e = \{v_{\lfloor x/2 \rfloor}, v_{\lfloor x/2 \rfloor+1}\}$ be the middle edge of the core path.
It is guaranteed that all vertices in $N^{r-1}(e)$ in $\HH$ are not already assigned output labels.
The partially labeled bipolar tree $\HH' = \labeling(\HH)$ is defined as the result of assigning output labels to vertices in $N^{r-1}(e)$ by the function $f$.\footnote{Note that the neighborhood function is evaluated w.r.t.~$H$.  In particular, the set $N^{r-1}(e)$ contains the vertices $v_{\lfloor x/2 \rfloor-r+1}, \ldots, v_{\floor{x/2} + r}$ of the core path,
and also contains parts of the trees $\TT_{\floor{x/2}-r+1}, \ldots, \TT_{\floor{x/2} + r}$.}
\item[\extend.]
Let $\HH=(\TT_1, \ldots, \TT_x)$ be a partially labeled bipolar tree with $x \in [\ell, 2w]$.
The partially labeled bipolar tree  $\HH' = \extend(\HH)$ is defined as follows.
Consider the decomposition $\HH = \XX \circ \YY \circ \ZZ$, where
$\YY = (\TT_{\lfloor x/2 \rfloor-r+1}, \ldots, \TT_{\lfloor x/2 \rfloor + r})$.
Then $\HH' = \pump(\XX,w) \circ \YY \circ \pump(\ZZ,w)$.
\end{description}

Intuitively, the goal of the operation $\extend$ is to extend the length of the core path of $\HH$ while preserving the type of $\HH$, due to Lemma~\ref{lem:replace-2}. Suppose that the number of vertices in the core path of $\HH$ is in the range $[\ell, 2\ell]$.  The prefix $\XX$ and suffix $\ZZ$ are stretched to lengths in the range $[w,w+\Lpump]$,
and the middle part $\YY$ has length $2r$, so the core path of $\HH'$ has length in the range
$[2(w+r), 2(w+r+\Lpump)]$.

The reason that the $\extend$ operation does not modify the middle part $\YY$ is to ensure that (given any labeling function $f$) the type of $\HH' =\extend(\labeling(\HH))$ is invariant over all choices of the parameter $w$.\footnote{Notice that $\extend$ is applied {\em after} $\labeling$. Thus, the vertices that are assigned output labels during $\labeling$ must be within the middle part $\YY$, no part of which is modified during $\extend$.}
We have the following lemma.

\begin{lemma}\label{lem:recover-1}
Let $\GG=(G,\LL)$ be a partially labeled graph and $\HH=(H,\LL)$ be a bipolar subtree of $\GG$ with poles $(s,t)$.
Let $\tilde{\HH}$ be another partially labeled bipolar tree with $\type(\tilde{\HH}) = \type(\HH)$
and $\HH' = \extend(\labeling(\tilde{\HH}))$.
If $\GG' = \replace(\GG,\HH,\HH')$ admits a legal labeling $\LL_\diamond'$,
then $\GG$ admits a legal labeling $\LL_\diamond$ such that ${\LL_\diamond}(v) = \LL_\diamond'(v')$ for each vertex $v \in V(G) - V(H)$ and its corresponding vertex $v' \in V(G') - V(H')$.
\end{lemma}

\begin{proof}
Recall that the operation $\extend$ guarantees that $\type(\extend(\tilde{\HH})) = \type(\tilde{\HH}) = \type(\HH)$.
Define $\HH'' = \extend(\tilde{\HH})$ and $\GG'' = \replace(\GG,\HH,\HH'')$.
Observe that $\HH' = \extend(\labeling(\tilde{\HH}))$ can be seen as the result of fixing the output labels of some unlabeled vertices in $\HH'' = \extend(\tilde{\HH})$. Therefore, $\LL_\diamond'$ is also a legal labeling of $\GG''$.
By Lemma~\ref{lem:replace}, the desired legal labeling  $\LL_\diamond$ of $\GG = \replace(\GG'',\HH'',\HH)$ can be obtained from the legal labeling  $\LL_\diamond'$ of $\GG''$.
\end{proof}

In addition to $\extend$ and $\labeling$, we also modify trees using the $\cut$ operation, defined below.
\begin{description}
\item[\cut.]
Let $\GG=(G,\LL)$ be a partially labeled graph and $\HH=(H, \LL)$ be a bipolar subtree with poles $(s,t)$.
Suppose that $\HH$ is connected to the rest of $\GG$ via two edges $\{u,s\}$ and $\{v,t\}$.
The partially labeled graph $\GG' = \cut(\GG,\HH)$
is formed by
(i) duplicating $\HH$ and the edges $\{u,s\},\{v,t\}$ so that $u$ and $v$ are attached to both copies of $\HH$,
(ii) removing the edge that connects $u$ to one copy of $\HH$, and removing the edge from $v$ to the other copy of $\HH$.
\end{description}

\begin{figure}
\centerline{\scalebox{.35}{\includegraphics{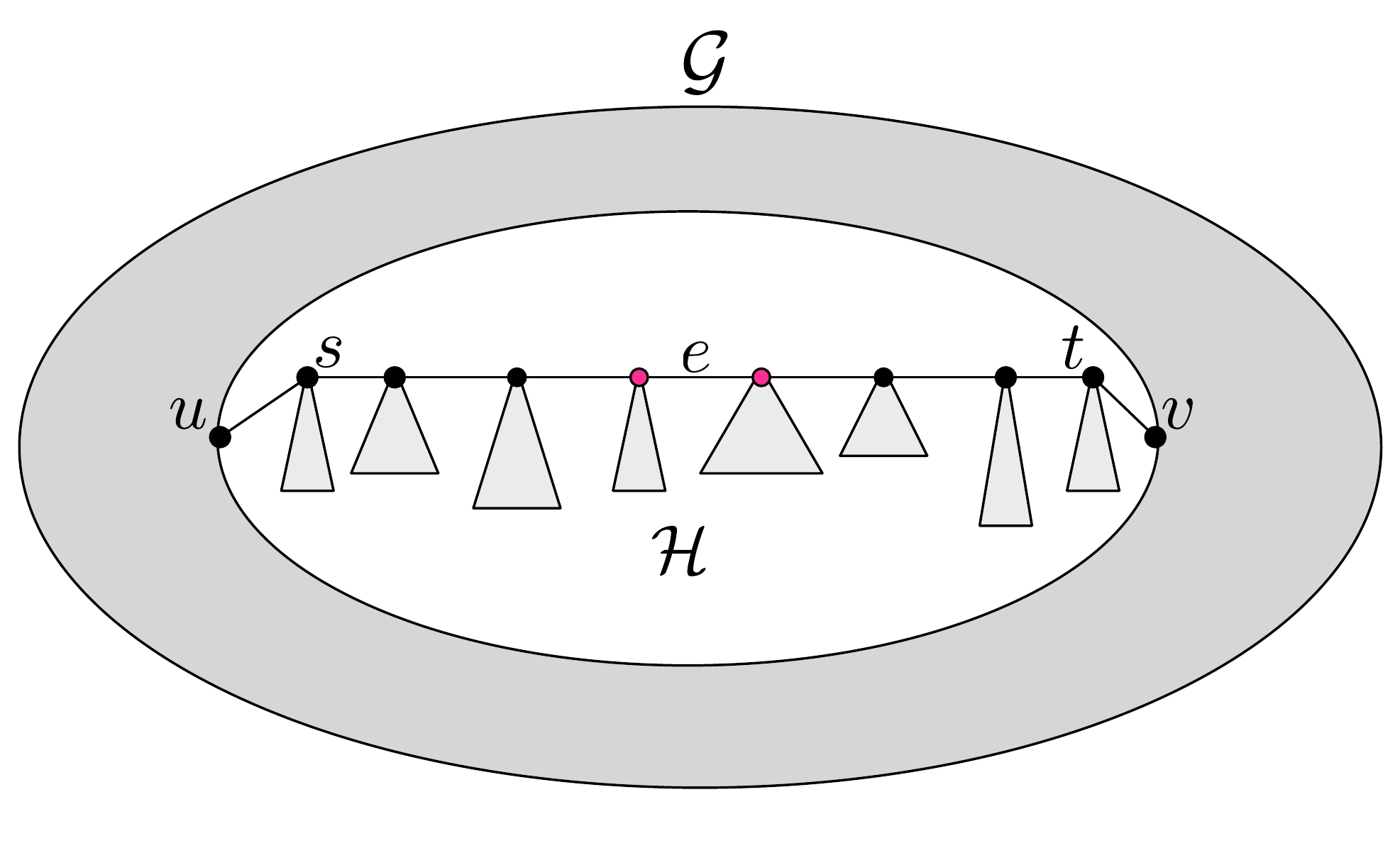}}\hspace{1cm}\scalebox{.35}{\includegraphics{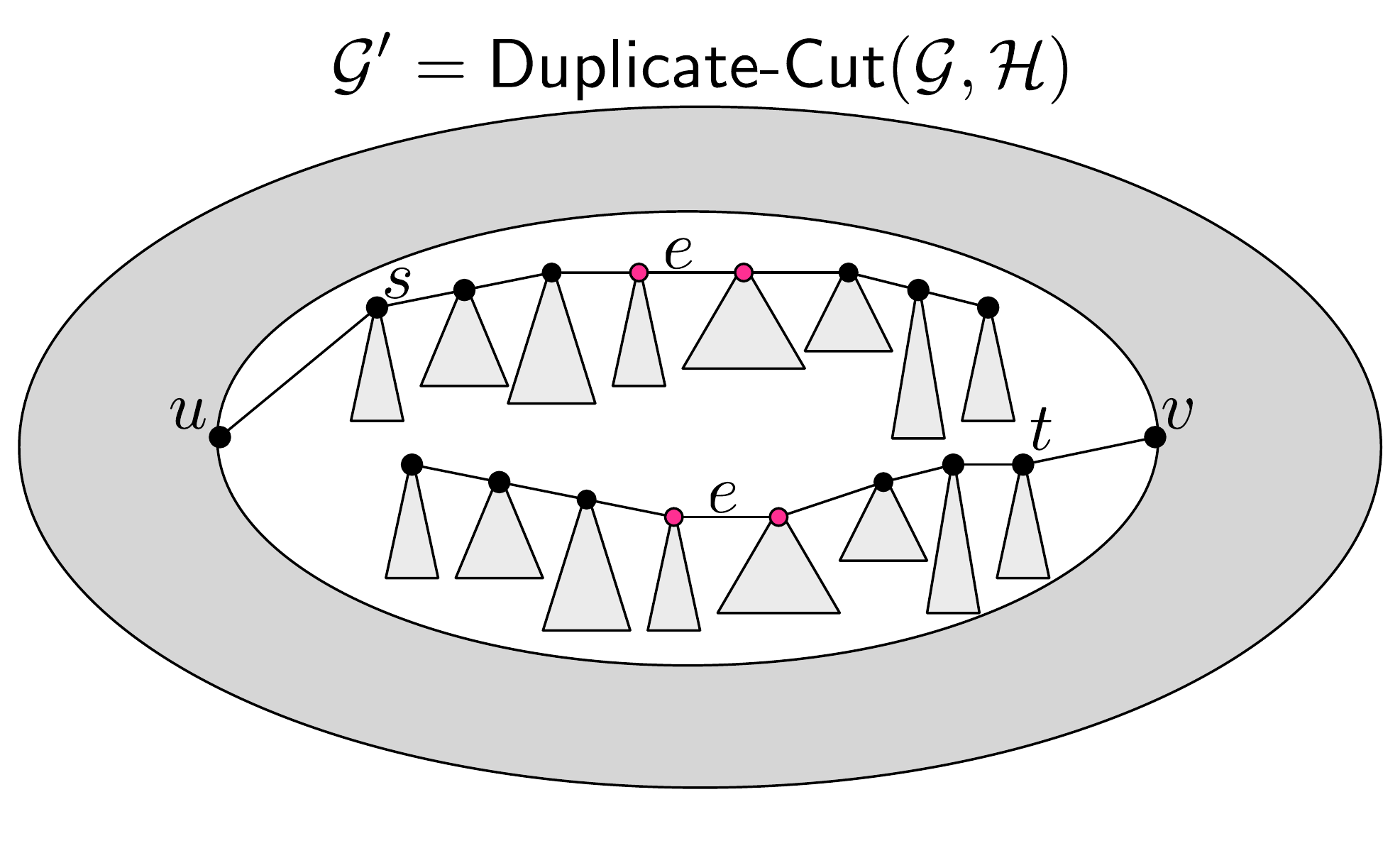}}}
\caption{\label{fig:cut}{\bf Left:} A bipolar subtree $\HH$ is attached to the rest of the graph $\GG$
via edges $\{u,s\},\{v,t\}$.  The pink nodes have been pre-committed to output labels by $\labeling$ ($r=1$).
{\bf Right:} The $\cut$ operation duplicates $\HH$ and attaches one copy to
$u$ and the other to $v$.}
\end{figure}

Later on we will see that both poles of a bipolar tree are responsible for computing the labeling of the tree.
On the other hand, we do not want the poles to have to communicate too much.
As Lemma~\ref{lem:recover-2} shows, the $\cut$ operation (in conjunction with $\extend$ and $\labeling$)
allows both poles to work independently and cleanly integrate their labelings afterward.

\begin{lemma}\label{lem:recover-2}
Let $\HH = \extend(\labeling(\tilde{\HH}))$ for some partially labeled bipolar tree  $\tilde{\HH}$.
If $\GG' = \cut(\GG,\HH)$ admits a legal labeling $\LL_\diamond'$, then $\GG$ admits a legal labeling $\LL_\diamond$ such that ${\LL_\diamond}(v) = \LL_\diamond'(v')$ for each vertex $v \in V(G) - V(H)$ and a particular
corresponding vertex $v'$ in $\GG'$.
\end{lemma}
\begin{proof}
Let $\GG'=(G', \LL')$.
We write $\HH = (\TT_1, \ldots, \TT_x)$.
Let $(v_1, \ldots, v_x)$ be the core path of $\HH$, where $s = v_1$ and $t = v_x$ are the two poles of $\HH$.
Let $\{u,s\}$ and $\{v,t\}$ be the two edges that connect $H$ two the rest of $G$.
Let $e=\{v_j,v_{j+1}\}$ be the edge in the core path of $\HH$ such that the output labels of vertices in $N^{r-1}(e)$ in $\HH$ were fixed by $\labeling$.\footnote{Because $\pump$ usually does not extend $\XX$ and $\ZZ$ by precisely the same amount, the edge $e$ is generally not \emph{exactly} in the middle.}
We write $\HH_u$ (resp., $\HH_v$) to denote the copy of $\HH$ in $\GG'$ that attaches to $u$ (resp., $v$).
Define a mapping $\phi$ from $V(G)$ to $V(G')$ as follows.
\begin{itemize}
\item For $z \in V(G) - V(H)$, $\phi(z)$ is the corresponding vertex in $G'$.
\item For $z \in \bigcup_{i=1}^{j} \TT_i$,  $\phi(z)$ is the corresponding vertex in $H_u$.
\item For $z \in \bigcup_{i=j+1}^{x} \TT_i$,  $\phi(z)$ is the corresponding vertex in $H_v$.
\end{itemize}
We set $\LL_\diamond(z) = \LL_\diamond'(\phi(z))$ for each $z \in V(G)$.
It is straightforward to verify that the distance-$r$ neighborhood view (with output labeling $\LL_\diamond$) of each vertex $z \in V(G)$ is the same as the distance-$r$ neighborhood view (with output labeling $\LL_\diamond'$) of its corresponding vertex $\phi(z)$ in $G'$. Thus, $\LL_\diamond$ is a legal labeling.
\end{proof}

Notice that in the proof of Lemma~\ref{lem:recover-2}, the only property of $\HH$ that we use is that $N^{r-1}(e)$ was assigned output labels in the application of $\labeling(\tilde{\HH})$.

\subsection{A Hierarchy of Partially Labeled Trees \label{sec.tree-construct}}

In this section we construct several sets
of partially labeled unipolar and bipolar trees---$\{\mathscr{T}_i\}$,
$\{\mathscr{H}_i\}$,
and
$\{\mathscr{H}_i^+\}$, $i\in\mathbb{Z}^+$---using the operations $\extend$ and $\labeling$.
If each member of ${\Tset}^\star = \bigcup_i \mathscr{T}_i$ admits a legal labeling, then
we can use these trees to design an $O(\log n)$-time $\DetLOCAL$ algorithm for $\mathcal{P}$.
Each $\TT \in {\Tset}^\star$ is partially labeled in the following restricted manner.
The tree $\TT=(T,\LL)$ has a set of {\em designated edges} such that $\LL(v)\neq\; \bottom$ is defined if and only if
$v \in N^{r-1}(e)$ for some designated edge $e$; these vertices were issued labels by some invocation of $\labeling$.

The sets of bipolar trees $\{\mathscr{H}_i\}_{i \in \mathds{Z}^{+}}$ and
$\{\mathscr{H}_{i}^+\}_{i \in \mathds{Z}^{+}}$
and unipolar trees $\{\mathscr{T}_i\}_{i \in \mathds{Z}^{+}}$
are defined inductively.
In the base case we have $\Tset_1 = \{\TT\}$, where $\TT$ is the unique unlabeled, single-vertex, unipolar tree.
\begin{description}
\item[$\;\;\;\;\Tset$ Sets:]
For each $i > 1$, $\Tset_i$ consists of all partially labeled rooted trees $\TT$ formed in the following manner.
The root $z$ of $\TT$ has degree $0 \leq \deg(z) \leq \Delta$.
Each child of $z$ is either
(i) the root of a partially labeled rooted tree $\TT'$ from $\Tset_{i-1}$ (having degree at most $\Delta-1$),
or
(ii) one of the two poles of a bipolar tree $\HH$ from ${\Hset}_{i-1}^{+}$.
\item[$\;\;\;\;\Hset$ Sets:]
For each $i \geq 1$, ${\Hset}_i$ contains all partially labeled bipolar trees $\HH = (\TT_j)_{j \in [x]}$
such that $x\in [\ell,2\ell]$,
and for each $j \in [x]$, $\TT_j \in {\Tset}_i$, where the root of $\TT_j$
has degree at most $\Delta-2$.
For example, since ${\Tset}_1$ contains only the single-vertex unlabeled tree,
$\Hset_1$ is the set of all bipolar, unlabeled paths with between $\ell$ and $2\ell$
vertices.

\item[$\;\;\;\;\Hset^+$ Sets:]
For each $i \geq 1$, ${\Hset}_i^+$ is constructed by the following procedure.
      If $i = 1$, initialize ${\Hset}_1^+ \leftarrow \emptyset$,
      otherwise initialize ${\Hset}_i^+ \leftarrow {\Hset}_{i-1}^+$.
Consider each $\HH \in {\Hset}_i$ in some canonical order.
If there does \emph{not} already
exist a partially labeled bipolar tree $\tilde{\HH}$ such that
$\type(\tilde{\HH}) = \type(\HH)$ and $\extend(\labeling(\tilde{\HH})) \in {\Hset}_i^+$,
then update  ${\Hset}_i^+ \leftarrow {\Hset}_i^+ \cup \{\extend(\labeling({\HH}))\}$.
\end{description}
%
Observe that whereas $\{\Tset_i\}$ and $\{\Hset_i\}$ grow without end, and contain arbitrarily large trees,
the cardinality of $\Hset_i^+$ is at most the total number of types, which is constant.\footnote{However,
it is not necessarily true that $\Hset_i^+$ contains at most one bipolar tree of each type.  The $\extend$
operation is type-preserving, but this is not true of $\labeling$:
$\type(\labeling(H))$ may not equal $\type(H)$, so it is  possible that $\Hset_i^+$
contains two members of the same type.}
This is due to the observation that whenever we add a new partially labeled bipolar tree
$\extend(\labeling({\HH}))$ to ${\Hset}_i^+$, it is guaranteed that there is no other
partially labeled bipolar tree
$\extend(\labeling(\tilde{\HH})) \in {\Hset}_i^+$ such that $\type(\tilde{\HH}) = \type(\HH)$.
The property that $|\Hset_i^+|$ is constant is crucial in the proof of Lemma~\ref{lem:func-decide}.
Lemmas~\ref{lem:subset}--\ref{lem:w-converge} reveal some useful properties of these sets.

\begin{lemma}\label{lem:subset}
We have
(i) $\mathscr{T}_1 \subseteq \mathscr{T}_2 \subseteq \cdots$,
(ii) $\mathscr{H}_1 \subseteq \mathscr{H}_2 \subseteq \cdots$, and
(iii) $\mathscr{H}_1^+ \subseteq \mathscr{H}_2^+ \subseteq \cdots$.
\end{lemma}
\begin{proof}
By construction, we already have $\mathscr{H}_1^+ \subseteq \mathscr{H}_2^+ \subseteq \cdots$.
Due to the construction of $\Hset_i$ from the set $\Tset_i$, it is guaranteed that if
$\Tset_j \subseteq \Tset_{j+1}$ holds then $\Hset_j \subseteq \Hset_{j+1}$ holds as well.
Thus, it suffices to show that $\Tset_1 \subseteq \Tset_2 \subseteq \cdots$.
This is proved by induction.

For the base case, we have $\Tset_1 \subseteq \Tset_2$ because $\Tset_2$ also contains $\TT\in\Tset_1$,
the unlabeled, single-vertex, unipolar tree.

For the inductive step, suppose that we already have $\Tset_1 \subseteq \Tset_2 \subseteq \cdots \subseteq \Tset_i$,
 $i \geq 2$.
Then we show that $\Tset_i \subseteq \Tset_{i+1}$. Observe that the set $\Tset_{i+1}$ contains all partially labeled rooted trees constructed by attaching  partially labeled trees from the sets $\Hset_{i}^+$ and $\Tset_{i}$ to the root vertex. We already know that $\Hset_{i-2}^+ \subseteq \Hset_{i-1}^+$, and by the inductive hypothesis we have $\Tset_{i-2} \subseteq \Tset_{i-1}$. Thus, each $\TT \in \Tset_{i}$ must also appear in the set $\Tset_{i+1}$.
\end{proof}

If $\Tset$ and $\Hset$ are arbitrary sets of unipolar and bipolar trees, we define
$\class(\Tset) = \{\class(\TT) \;|\; \TT\in \Tset\}$ and $\type(\Hset) = \{\type(\HH) \;|\; \HH\in\Hset\}$
to be the set of classes and types appearing among them.

\begin{lemma}\label{lem:k-converge}
Define $k^\star = |\mathscr{C}|$, where $\mathscr{C}$ is the set of all classes.
Then $\class(\mathscr{T}^\star) = \class(\mathscr{T}_{k^\star})$.
\end{lemma}
\begin{proof}
For each $i > 1$, $\class({\Tset}_i)$ depends only on  $\type({\Hset}_{i-1}^+)$ and $\class(\Tset_{i-1})$,
due to Lemmas~\ref{lem:replace-rootedtree} and~\ref{lem:type2class}.
Let $i^\ast$ be the smallest index such that  $\class({\Tset}_{i^\ast}) = \class({\Tset}_{i^\ast + 1})$.
Then we have $\type({\Hset}_{i^\ast}) =\type({\Hset}_{i^\ast + 1})$
and as a consequence,
${\Hset}_{i^\ast}^+ = {\Hset}_{i^\ast + 1}^+$.
This implies that $\class({\Tset}_{i^\ast + 1}) = \class({\Tset}_{i^\ast + 2})$. By repeating the same argument, we conclude that for each $j \geq i^\ast$, we have $\class({\Tset}_{j}) = \class({\Tset}_{i^\ast}) = \class({\Tset}^\star)$.
Since $\mathscr{T}_1 \subseteq \mathscr{T}_2 \subseteq \ldots$ (Lemma~\ref{lem:subset}),
we have $i^\ast \leq |\mathscr{C}|$.
\end{proof}

\begin{lemma}\label{lem:w-converge}
For each $i$,
$\class(\Tset_i)$ does not depend on the parameter $w$ used in the operation $\extend$.
\end{lemma}
\begin{proof}
Let $\HH=(\TT_1, \ldots, \TT_x)$ be any partially labeled bipolar tree with $x \geq 2r+2\Lpump$.
The type of $\HH' = \extend(\HH)$ is invariant over all choices of the parameter $w$.
Thus, by induction, the sets $\class(\Tset_i)$, $\type(\Hset_i)$, and $\type(\Hset_i^+)$
are also invariant over the choice of $w$.
\end{proof}

\paragraph{Feasible Labeling Function.} In view of Lemma~\ref{lem:w-converge}, $\class(\Tset^\star)$ depends only on the choice of the labeling
function $f$ used by $\labeling$. We call a function $f$ {\em feasible} if implementing $\labeling$ with $f$ makes
each tree in $\class(\Tset^\star)$ good, i.e., its partial labeling can be extended to a complete and legal labeling.
In Section~\ref{sec.tree} we show that
given a feasible function, we can generate a $\DetLOCAL$ algorithm to solve $\mathcal{P}$ in $O(\log n)$-time.
In Section~\ref{sect:feasible-functions}, we show that (i) a feasible function can be derived from any given an $n^{o(1)}$-time $\RandLOCAL$ algorithm for $\mathcal{P}$, and (ii) the existence of a feasible function is decidable. These results together imply the $\omega(\log n)$---$n^{o(1)}$ gap. Moreover, given an LCL problem $\mathcal{P}$ on bounded degree trees, it is decidable whether the $\RandLOCAL$ complexity of $\mathcal{P}$ is $n^{\Omega(1)}$ or the $\DetLOCAL$ complexity of $\mathcal{P}$ is $O(\log n)$.

\subsection{A $O(\log n)$-time $\DetLOCAL$ Algorithm from a Feasible Labeling Function\label{sec.tree}}

In this section, we show that given a feasible function $f$ for the LCL problem $\mathcal{P}$,
it is possible to design a $O(\log n)$-time $\DetLOCAL$ algorithm for
$\mathcal{P}$ on bounded degree trees.

Regardless of $f$, the algorithm begins by computing the
graph decomposition $V(G) = V_1\cup\cdots\cup V_L$, with $L = O(\log n)$; see Section~\ref{sec.decomp}.
We let the three infinite sequences  $\{\mathscr{H}_i\}_{i \in \mathds{Z}^{+}}$,
$\{\mathscr{H}_{i}^+\}_{i \in \mathds{Z}^{+}}$, and $\{\mathscr{T}_i\}_{i \in \mathds{Z}^{+}}$ be constructed
with respect to the feasible $f$ and any parameter $w$.

\paragraph{A Sequence of Partially Labeled Graphs.}
We define below a sequence of partially labeled graphs $\RR_1, \RR_2, \ldots, \RR_L$,
where $\RR_1$ is the unlabeled tree $G$ (the underlying communications network),
and $\RR_{i+1}$ is constructed from $\RR_i$ using the graph operations $\extend,\labeling,$ and $\cut$.
An alternative, and helpful way to visualize $\RR_i$ is that it is obtained by stripping away some vertices of $G$,
and then grafting on some \emph{imaginary} subtrees to its remaining vertices.
Formally, the graph $\RR_i$ is formed by taking ${G}_{i}$ (the subforest induced by $\bigcup_{j=i}^L V_j$, defined in Section~\ref{sec.decomp}),
and identifying each vertex $u\in V(G_i)$ with the root of a
partially labeled imaginary tree $\TT_{u,i}\in \Tset_i$ (defined within the proof of Lemma~\ref{lem:GG}).
Since $G_L$ consists solely of isolated vertices,
$\RR_L$ is the disjoint union of trees drawn from $\Tset_L$.

Once each vertex $v \in V(G_i) = \bigcup_{j=i}^L V_j$ in the communication network $G$ knows $\TT_{v,i}$, we are able to simulate the imaginary graph $\RR_i$ in the communication network $G$. In particular, a \emph{legal labeling of $\RR_i$} is represented by storing the entire output labeling of the (imaginary)
tree $\TT_{v,i}$ at the (real) vertex $v \in V(G_i)$.

The official, inductive construction of $\RR_i$ is described
in the proof of Lemma~\ref{lem:GG}.

\begin{lemma}\label{lem:GG}
Suppose that a feasible function $f$ is given.
The partially labeled graphs $\RR_1, \ldots, \RR_L$
and partially labeled trees $\{\TT_{v,i} \;|\; v\in V(G_i),\, i\in[L]\}$
can be constructed in $O(\log n)$ time meeting the following conditions.
\begin{enumerate}
\item For each $i\in [1,L]$, each vertex $v \in V(G_i) = \bigcup_{j=i}^L V_j$ knows $\TT_{v,i} \in {\Tset}_i$.
%
\item For each $i \in [2,L]$, given a legal labeling of $\RR_i$,
a legal labeling of $\RR_{i-1}$ can be computed in $O(1)$ time.
\end{enumerate}
\end{lemma}

\begin{proof}
Part (1) of the lemma is proved by induction.

\paragraph{Base Case.}
Define $\RR_1 = G$.  This satisfies the lemma since
$\TT_{v,1} \in\Tset_1$ must be the unlabeled single-vertex tree, for each $v \in V(G)$.

\paragraph{Inductive Step.}
We can assume inductively that $\RR_{i-1}$ and $\{\TT_{v,i-1} \;|\; v\in V(G_{i-1})\}$ have been defined
and satisfy the lemma.
The set ${\Pset}_{i-1}$ was defined in Section~\ref{sec.decomp}.
Each $P \in \Pset_{i-1}$ is a path such that $\deg_{\Gd{i-1}}(v) = 2$ for each vertex $v \in V(P)$
and $|V(P)| \in [\ell, 2 \ell]$.
Fix a path $P=(v_1, \ldots, v_x) \in \Pset_{i-1}$.  The bipolar graphs $\HH_P$ and $\HH_P^+$ are
defined as follows.
\begin{itemize}
\item Define $\HH_P$ to be the partially labelled bipolar tree $(\TT_{v_1,i-1}, \ldots, \TT_{v_x,i-1})$
Notice that $\HH_P$ is a subgraph of $\RR_{i-1}$.
Since $\TT_{v_j,i-1} \in \Tset_{i-1}$, for each $j\in [x]$, it follows that $\HH_P \in {\Hset}_{i-1}$.

\item Construct $\HH_P^+$ as follows. Select the unique member $\tilde{\HH} \in \Hset_{i-1}$
such that $\type(\tilde{\HH})=\type(\HH_P)$ and let $\HH_P^+ = \extend(\labeling(\tilde{\HH})) \in {\Hset}_{i-1}^+$.
Due to the definition of ${\Hset}_{i-1}^+$, such an $\tilde{\HH}$ must exist, since $\HH_P \in \Hset_{i-1}$.
\end{itemize}
The partially labeled graph $\RR_i$ is constructed from $\RR_{i-1}$ with the following three-step procedure.
See Figure~\ref{fig:RR} for a schematic example of how these steps work.
\begin{description}
\item[Step 1.] Define $\RR_{i-1}'$ as the result of applying the following operations on $\RR_{i-1}$. For each $v \in V_{i-1}$ such that $\TT_{v,i-1}$ is a connected component of $\RR_{i-1}$, remove $\TT_{v,i-1}$. Notice that a tree $\TT_{v,i-1}$ is a connected component of $\RR_{i-1}$ if and only if $v$'s neighborhood in $G$ contains only vertices at lower levels: $V_1,\ldots,V_{i-2}$.
\item[Step 2.] Define $\RR_{i-1}^+$ by the following procedure. (i) Initialize $\tilde{\GG} \leftarrow \RR_{i-1}'$. (ii) For each $P \in {\Pset}_{i-1}$, do $\tilde{\GG} \leftarrow \replace(\tilde{\GG}, \HH_P, \HH_P^+)$. (iii) Set $\RR_{i-1}^+ \leftarrow \tilde{\GG}$.
\item[Step 3.] Define $\RR_{i}$ by the following procedure. (i) Initialize $\tilde{\GG} \leftarrow \RR_{i-1}^+$. (ii) For each $P \in {\Pset}_{i-1}$, do $\tilde{\GG} \leftarrow \cut(\tilde{\GG},\HH_P^+)$. (iii) Set $\RR_{i} \leftarrow \tilde{\GG}$.
\end{description}

After Steps 1--3, for $v\in V(G_i)$, $\TT_{v,i}$ is now defined to be the tree in $\RR_i - (V(G_i) - \{v\})$ rooted at $v$. Notice that the two copies of $\HH_P^+$ generated during Step 3(ii) becomes subtrees of $\TT_{u,i}$ and $\TT_{v,i}$, where $u$ and $v$ are the two vertices in $V(G_i)$ adjacent to the two endpoints of $P$ in the graph $G$.


\begin{figure}
\centering
\scalebox{.43}{\includegraphics{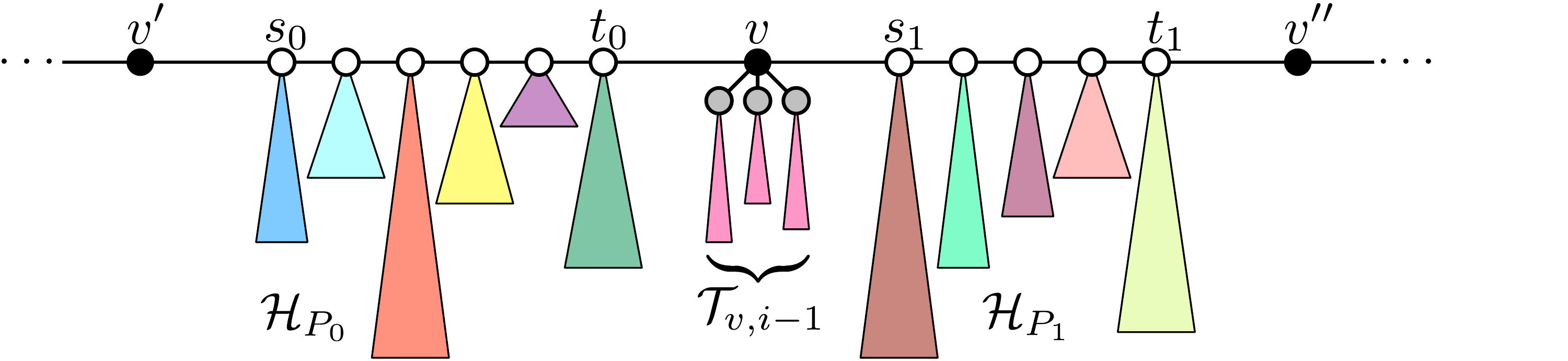}}\\

\vspace{1.5cm}

\scalebox{.43}{\includegraphics{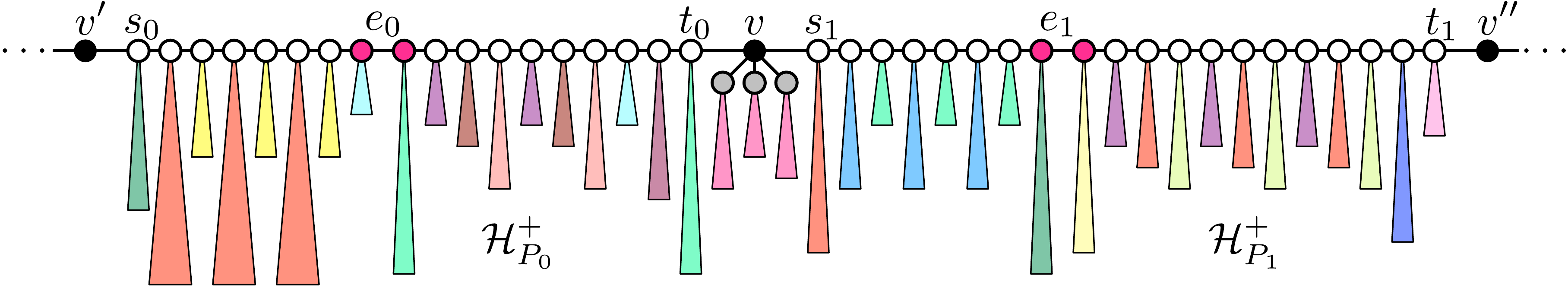}}\\

\vspace{1.5cm}
\scalebox{.43}{\includegraphics{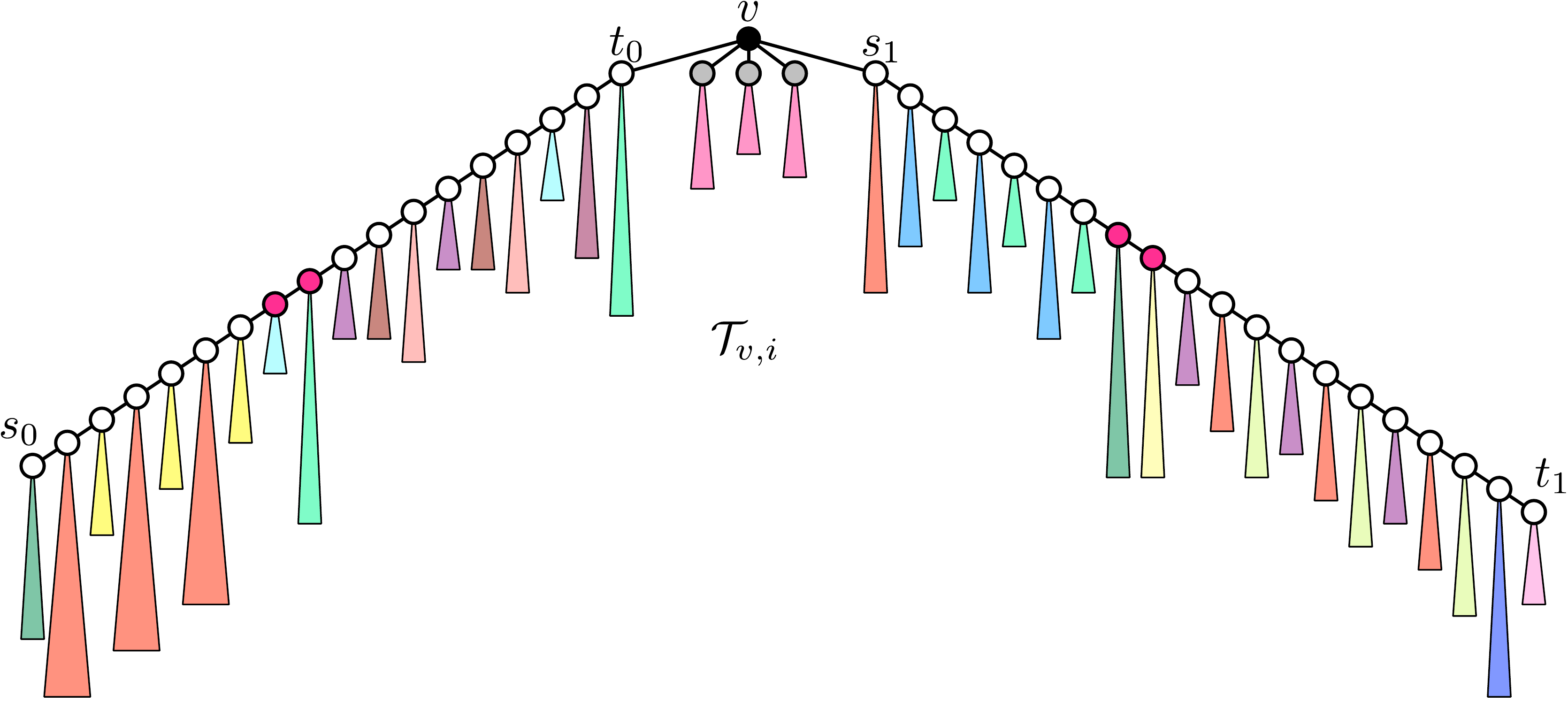}}
\caption{\label{fig:RR}%
{\bf Top:}
In this example $v$ was a vertex in a long degree-2 path tagged $(i-1)_C$ by the decomposition procedure,
and subsequently promoted to $V_i$.
Black vertices are in $V_i$ (or above); white vertices are in $V_{i-1}$; gray vertices are in $V_{i-2}$ or below.
The paths $P_0 = (s_0,\ldots,t_0)$ and $P_1=(s_1,\ldots,t_1)$ adjacent to $v$ have constant length, between $\ell$ and $2\ell$.
The colored subtrees grafted onto white and gray vertices are imaginary subtrees formed in the construction of $\RR_{i-1}$.
{\bf Middle:}  The graph is transformed by finding the graph $\tilde{\HH}_b \in \Hset_{i-1}^+$, $b\in\{0,1\}$
that has the same type as $\HH_{P_b}$, and replacing $\HH_{P_b}$ with
$\HH_{P_b}^+ = \extend(\labeling(\tilde{\HH}_b))$.  The vertices receiving
pre-committed labels are indicated in pink ($r=1$).
{\bf Bottom:} We duplicate $\HH_{P_b}^+$, $b\in \{0,1\}$, and attach one of the copies of each duplicate to $v$.
(The copies of $\HH_{P_b}^+$ attached to $v',v''$ are not shown.)  The tree $\TT_{v,i}$ is the resulting tree rooted at $v$.
Since each subtree of $v$ is in $\Tset_{i-1}$ or $\Hset_{i-1}^+$, it follows that $\TT_{v,i} \in \Tset_i$.
In this case $v$ had no neighbors at higher levels ($i+1$ and above), so
$\TT_{v,i}$ is a connected component of $\RR_i$.   Thus, $v$ can locally compute a legal labelling of $\TT_{v,i}$.
}
\end{figure}

We now need to verify that $\RR_i$ satisfies all the claims of the lemma.
Given the partially labeled graph $\RR_i$, the partially labeled trees $\TT_{v,i}$ for all $v \in V(G_i)$ are uniquely determined.
According to the construction of $\RR_i$, each connected component of
$\RR_i - V(G_i)$ must be an imaginary tree that is either
(i) some $\TT_{v,j}$, where $v \in V_j$ and $j \in \{1, \ldots, i-1\}$ or
(ii) a copy of $\HH_P^+$, where $P \in {\Pset}_{j}$ and $j \in \{1, \ldots, i-1\}$.
By induction (and Lemma~\ref{lem:subset}), for $v \in V_1  \cup\cdots \cup V_{j}$ and $j \in \{1, \ldots, i-1\}$, we have $\TT_{v,j} \in {\Tset}_j \subseteq {\Tset}_{i-1}$; for each $P \in {\Pset}_{j}$ where $j \in \{1, \ldots, i-1\}$, we have $\HH_P^+ \in \Hset_{j}^+ \subseteq \Hset_{i-1}^+$.
According to the inductive definition of ${\Tset}_i$,
for each $v \in V(G_i)$ we have $\TT_{v,i} \in {\Tset}_i$. This concludes the induction of Part (1).\\

We now turn to the proof of Part (2) of the lemma.
Suppose that we have a legal labeling of $\RR_i$, where the labeling of $\TT_{v,i}$ is stored in $v \in V(G_i)$.
We show how to compute a legal labeling of $\RR_{i-1}$ in $O(1)$ time as follows.
Starting with any legal labeling $\LL_1$ of $\RR_i$, we compute a legal labeling $\LL_2$ of $\RR_{i-1}^+$, a legal labeling $\LL_3$ of $\RR_{i-1}'$, and finally a legal labeling $\LL_4$ of $\RR_{i-1}$.
Throughout the process, the labels of all vertices in $\bigcup_{j=i}^L V_j$ are stable under $\LL_1, \LL_2, \LL_3$, and $\LL_4$.
Recall that $\RR_i$, $\RR_{i-1}^+$, $\RR_{i-1}'$, and $\RR_{i-1}$ are all \emph{imaginary}. ``Time'' refers to communications rounds in the \emph{actual} network $G$, not any imaginary graph.
\begin{description}
\item[From $\LL_1$ to $\LL_2$.]
Let $s,t$ be the poles of $\HH_P^+$ and $u,v$ be the vertices outside of $\HH_P^+$ in $\RR_{i-1}^+$ adjacent to $s,t$, respectively.
At this point $u$ and $v$ have legal labelings of $\TT_{u,i}$ and $\TT_{v,i}$, both trees of which contain a copy of $\HH_P^+$.
Using Lemma~\ref{lem:recover-2} we integrate the labelings of $\TT_{u,i}$ and $\TT_{v,i}$ to
fix a single legal labeling $\LL_2$ of $\HH_P^+$ in $\RR_{i-1}^+$.\footnote{It is not necessary to physically store the
entire $\LL_2$ on $\HH_P^+$.  To implement the following steps, it suffices that
$s,t$ both know what $\LL_2$ is on the subgraph induced by the $(2r-1)$-neighborhood of $\{s,t\}$ in $\HH_P^+$.}
\item[From $\LL_2$ to $\LL_3$.]
A legal labeling $\LL_3$ of $\RR_{i-1}'$ is obtained by applying Lemma~\ref{lem:recover-1}.
For each $P \in {\Pset}_{i-1}$, the labeling $\LL_3$ on $\HH_P$ in $\RR_{i-1}'$
can be determined from the labeling $\LL_2$  of $\HH_P^+$ in $\RR_{i-1}^+$.
In greater detail, suppose $s,t$ are the poles of $\HH_P/\HH_P^+$, which know
$\LL_2$ on the $(2r-1)$-neighborhood of $\{s,t\}$ in $\HH_P^+$.
By Lemma~\ref{lem:recover-1}, there exists a legal labeling $\LL_3$ on $\HH_P$,
which can be succinctly encoded by fixing $\LL_3$ on the $(2r-1)$-neighborhoods
of the \emph{roots} of each unipolar tree on the core path $(s=v_1,\ldots,v_x=t)$ of $\HH_P$.
Thus, once $s,t$ calculate $\LL_3$, they can transmit the relevant information with constant-length messages
to the roots $v_1,\ldots,v_x$.  At this point each $v_j \in V_{i-1}$ can locally compute an extension of its labeling
to all of $\TT_{v_j,i-1}$.
\item[From $\LL_3$ to $\LL_4$.]
Notice that $\RR_{i-1}$ is simply the disjoint union of $\RR_{i-1}'$---for which we already have a legal labeling $\LL_3$---and
each $\TT_{v,i-1}$ that is a connected component of $\RR_{i-1}$.
A legal labeling $\LL_4$ of $\TT_{v,i-1}$ is computed locally at $v$,
which is guaranteed to exist since $\TT_{v,i-1}\in \Tset_{i-1}$.
\end{description}
This concludes the proof of the lemma.
\end{proof}

\begin{lemma}\label{lem:lognalg}
Let $\mathcal{P}$ be any LCL problem on trees with $\Delta = O(1)$.
Given a feasible function $f$,
the LCL problem $\mathcal{P}$ can be solved in $O(\log n)$ time in $\DetLOCAL$.
\end{lemma}
\begin{proof}
First compute a graph decomposition in $O(\log n)$ time.
Given the graph decomposition, for each $i \in [L]$, each vertex $v \in V_i$ computes the partially labeled rooted trees
$\TT_{v,j}$ for all $j \in [i]$; this can be done in $O(\log n)$ rounds.
Since $f$ is feasible, each partially labeled tree in $\Tset^\star$ admits a legal labeling.
Therefore, $\RR_{L}$ admits a legal labeling, and such a legal labeling can be computed without communication by the vertices in $V_L$.
Starting with any legal labeling of $\RR_{L}$,
legal labelings of $\RR_{L-1},\ldots,\RR_1=G$ can be computed in $O(\log n)$ additional time,
using Lemma~\ref{lem:GG}(2).
\end{proof}

\subsection{Existence of Feasible Labeling Function\label{sect:feasible-functions}}

In Lemmas \ref{lem:func-exist} and~\ref{lem:func-decide}
we show \emph{two} distinct ways to arrive at a feasible labeling function.
In Lemma~\ref{lem:func-exist} we assume that we are given the code of
a $\RandLOCAL$ algorithm $\mathcal{A}$ that solves
$\mathcal{P}$ in $n^{o(1)}$ time with at most $1/n$ probability of failure.
Using $\mathcal{A}$ we can extract a feasible labeling function $f$.\footnote{The precise
running time of $\mathcal{A}$ influences the $w$ parameter used by $\extend$.
For example, if $\mathcal{A}$ runs in $O(\log^2 n)$ time then $w$ will be smaller
than if $\mathcal{A}$ runs in $n^{1/\log\log\log n}$ time.}
Lemma~\ref{lem:func-exist} suffices to prove our $n^{o(1)}\rightarrow O(\log n)$
speedup theorem but, because it needs the code of $\mathcal{A}$, it is
insufficient to answer a more basic question.  Given the description of an LCL $\mathcal{P}$,
is $\mathcal{P}$ solvable in $O(\log n)$ time on trees or not?
Lemma~\ref{lem:func-decide} proves that this question is, in fact, decidable,
which serves to highlight the delicate boundary between decidable and undecidable
problems in LCL complexity~\cite{Brandt+17,NaorS95}.

\begin{lemma}\label{lem:func-exist}
Suppose that there exists a $\RandLOCAL$ algorithm $\mathcal{A}$ that solves $\mathcal{P}$ in $n^{o(1)}$ time
on $n$-vertex bounded degree trees, with local probability of failure at most $1/n$.
Then there exists a feasible function $f$.
\end{lemma}

\begin{proof}
Define $\beta = |\LabelOut|^{\Delta^{r}}$ to be an upper bound on the number of distinct output labelings of $N^{r-1}(e)$, where $e$ is any edge in any graph of maximum degree $\Delta$.
Define $N$ as the maximum number of vertices of a tree in ${\Tset}_{k^\star}$
over {\em all} choices of labeling function $f$.
As $\Delta, r, k^\star$ are all constants, we have $N = w^{O(1)}$.
Define $t$ to be the running time of $\mathcal{A}$ on a $(\beta N + 1)$-vertex tree.
Notice that $t$ depends on $N$, which depends on $w$.

\paragraph{Choices of $w$ and $f$.} We select $w$ to be sufficiently large such that
$w \geq 4(r+t)$. Such a $w$ exists since $\mathcal{A}$ runs in $n^{o(1)}$ time on an $n$-vertex graph,
and in our case $n$ is polynomial in $w$.
By our choice of $w$, the labeled parts of $\TT=(T, \LL) \in {\Tset}_{k^\star}$ are spread far apart.
In particular,
(i) the sets $N^{(r-1)+t}(e)$ for all designated edges $e$ in $\TT$ are disjoint,
(ii) for each vertex $v \in V(T)$, there is at most one designated edge $e$ such that the set
$N^{r+t}(v)$ intersects $N^{r-1+t}(e)$.

Let the function $f$ be defined as follows. Take any bipolar tree $\HH=(H, \LL')$ with middle edge $e$ on its core path.   The output labels of $N^{r-1}(e)$ are assigned by selecting the
{\em most probable labeling} that occurs when running $\mathcal{A}$ on the tree $\HH' = \extend(\HH)$,
while pretending that the underlying graph has $\beta N + 1$ vertices.
Notice that the most probable labeling occurs with probability at least $1/\beta$.

\paragraph{Proof Idea.}  In what follows, consider {\em any} partially labeled rooted tree $\TT=(T, \LL) \in {\Tset}_{k^\star}$, where the set ${\Tset}_{k^\star}$ is constructed with the parameter $w$ and function $f$.
All we need to prove is that $\TT$ admits a legal labeling $\LL_\diamond$.
Suppose that we execute $\mathcal{A}$ on $T$ while pretending that the total number of vertices is $\beta N + 1$.
Let $v$ be any vertex in $T$.
According to $\mathcal{A}$'s specs,
the probability that the output labeling of $N^r(v)$ is inconsistent with $\mathcal{P}$ is at most $1/ (\beta N + 1)$. However, it is not guaranteed that the output labeling resulting from $\mathcal{A}$ is also consistent with $\TT$, since $\TT$ is partially labeled. To handle the partial labeling of $\TT$, our strategy is to consider a modified distribution of random bits generated by vertices in $T$ that forces any execution of $\mathcal{A}$
to agree with $\LL$, wherever it is defined.
We will later see that with an appropriately chosen distribution of random bits, the outcome of $\mathcal{A}$ is a legal labeling of $\TT$ with positive probability.

\paragraph{Modified Distribution of Random Bits.} Suppose that an execution of $\mathcal{A}$ on a  $(\beta N + 1)$-vertex graph needs a $b$-bit  random string for each vertex.
For each designated edge $e$, let $U_e$ be the set of all assignments of $b$-bit strings to vertices in $N^{(r-1)+t}(e)$.
Define $S_e$ as the subset of $U_e$ such that $\rho \in S_e$ if and only if the following is true. Suppose that the $b$-bit string of each $u \in N^{(r-1)+t}(e)$ is $\rho(u)$. Using the $b$-bit string $\rho(u)$ for each $u \in N^{(r-1)+t}(e)$, the output labeling of the vertices in $N^{r-1}(e)$ resulting from executing $\mathcal{A}$  is the same as the output labeling specified by $\LL$. According to our choice of $f$, we must have $|S_e| / |U_e| \geq 1 / \beta$.

Define the  modified distribution $\mathcal{D}$ of $b$-bit random strings to the vertices in $T$ as follows. For each  designated edge $e$, the $b$-bit strings of the vertices in $N^{(r-1)+t}(e)$ are chosen uniformly at random from the set $S_e$. For the remaining vertices, their $b$-bit strings are chosen uniformly at random.

\paragraph{Legal Labeling $\LL_\diamond$ Exists.} Suppose that $\mathcal{A}$ is executed on $T$ with the modified distribution of random bits $\mathcal{D}$. Then it is guaranteed that $\mathcal{A}$ outputs a complete
labeling that is consistent with $\TT$.
Of course, the probability that $\mathcal{A}$ outputs an \emph{illegal} labeling under $\mathcal{D}$ may be different.
We need to show that $\mathcal{A}$ nonetheless succeeds with non-zero probability.

Consider any vertex $v \in V(T)$. The probability that $N^r(v)$ is inconsistent with $\mathcal{P}$
is at most $\beta / (\beta N + 1)$ under distribution $\mathcal{D}$, as explained below.
Due to our choice of $w$, the set $N^{r+t}(v)$ intersects at most one set $N^{r-1+t}(e)$ where $e$ is a designated edge. Let $U_v$ be the set of all assignments of $b$-bit strings to vertices in $N^{r+t}(v)$.
For  each $\rho \in U_v$, the probability that $\rho$ occurs in an execution of
$\mathcal{A}$  is $1/|U_v|$ if all random bits are chosen uniformly at random, and is at most $\beta/|U_v|$ under
$\mathcal{D}$.  Thus, the probability that $\mathcal{A}$ (using distribution $\mathcal{D}$)
labels $N^r(v)$ incorrectly is at most $\beta / (\beta N + 1)$.
The total number of vertices in $T$ is at most $N$.  Thus,
by the union bound, the probability that the output labeling of $\mathcal{A}$ (using $\mathcal{D}$)
is not a legal labeling is $\beta N / (\beta N + 1) < 1$.
\end{proof}

\begin{lemma}\label{lem:func-decide}
Given an LCL problem $\mathcal{P}$ on bounded degree graphs,
it is decidable whether there exists a feasible function $f$.
\end{lemma}
\begin{proof}
Throughout the construction of the three infinite sequences  $\{\mathscr{H}_i\}_{i \in \mathds{Z}^{+}}$,
$\{\mathscr{H}_{i}^+\}_{i \in \mathds{Z}^{+}}$, and $\{\mathscr{T}_i\}_{i \in \mathds{Z}^{+}}$, the number of
distinct applications of the operation $\labeling$ is constant,
as $|\Hset_i^+|$ is at most the total number of types.

Therefore, the number of distinct candidate functions $f$ that need to be examined is finite.
For each candidate labeling function $f$ (with any parameter $w \geq \ell$), in bounded amount of time we can construct the set $\mathscr{T}_{k^\star}$, as $k^\star = |\mathscr{C}|$ is a constant.
By examining the classes of the partially labeled rooted trees in $\mathscr{T}_{k^\star}$ we can infer whether the function $f$ is feasible (Lemma~\ref{lem:k-converge}).
Thus, deciding whether there exists a feasible function $f$ can be done in  bounded amount of time.
\end{proof}

Combining Lemmas~\ref{lem:lognalg}, \ref{lem:func-exist}, and \ref{lem:func-decide},
we obtain the following theorem.

\begin{theorem}\label{thm:tree-main}
Let $\mathcal{P}$ be any LCL problem on trees with $\Delta = O(1)$.
If there exists a $\RandLOCAL$ algorithm $\mathcal{A}$ that solves $\mathcal{P}$ in $n^{o(1)}$ rounds,
then there exists a $\DetLOCAL$ algorithm $\mathcal{A}'$ that solves $\mathcal{P}$ in $O(\log n)$ rounds.
Moreover, given a description of $\mathcal{P}$,
it is decidable whether the $\RandLOCAL$ complexity of $\mathcal{P}$ is
$n^{\Omega(1)}$ or the $\DetLOCAL$ complexity of $\mathcal{P}$ is $O(\log n)$.
\end{theorem}

\begin{remark}
Observe that our graph decomposition algorithm also works on graphs of girth at least $c \log n$, where $c$ is a sufficiently large constant depending on $\mathcal{P}$.
This implies that Theorem~\ref{thm:tree-main} also applies to the class of $n$-vertex
graphs with girth $\omega(\log n)$.
\end{remark}

\ignore{
\begin{remark}
Since an $O(\log n)$-time $\DetLOCAL$ algorithm can be derived from a feasible function $f$ with  any choice of target length $w$, the $\extend$ operation is unnecessary in the proof of Lemma~\ref{lem:lognalg}.
The reader may wonder what the role of $\extend$ is in the derivation of Theorem~\ref{thm:tree-main}.
The proof of Lemma~\ref{lem:func-exist} basically says that there exists a function $f$ and a parameter $w$ such that all classes in $\class(\Tset^\star)$ are good, given that $\mathcal{P}$ can be solved in $n^{o(1)}$ time in $\RandLOCAL$.
Therefore, without the $\extend$ operation, the $\omega(\log n)$---$n^{o(1)}$ gap can still be derived (from Lemma~\ref{lem:func-exist} and Lemma~\ref{lem:lognalg}).

The role of the $\extend$ operation is essential in the derivation of the decidability result. The $\extend$ operation comes into the play in the proof of Lemma~\ref{lem:w-converge}, which  says that if we can make all classes in  $\class(\Tset^\star)$ good for some (large) $w$, then any choice of target length $w \geq \ell = 2(r+\Lpump)$ also works, and so we do not need to worry about the unbounded number of possible values of $w$.
\end{remark}
}

\newcommand{\vbl}{\operatorname{vbl}}
\newcommand{\Alg}{\mathcal{A}}

\section{A Gap in the $\RandLOCAL$ Complexity Hierarchy}\label{sect:LLL}

Consider a set $\mathcal{V}$ of independent random variables, and a set $\mathcal{X}$ of \emph{bad events},
where $A\in \mathcal{X}$ depends only on some subset $\vbl(A) \subset \mathcal{V}$ of variables.\footnote{Each variable
$V\in \mathcal{V}$ may have a different distribution and range, so long as the range is some finite set.}
The \emph{dependency graph} $G_{\mathcal{X}} = (\mathcal{X}, \{(A,B) \;|\; \vbl(A) \cap \vbl(B) \neq \emptyset\})$ joins
events by an edge if they depend on at least one common variable.  The \emph{\Lovasz{} local lemma} (LLL) and its variants
give criteria under which $\Pr(\bigcap_{A\in \mathcal{X}} \overline{A}) > 0$, i.e., it is possible that all bad events do not occur.  We will narrow
our discussion to \emph{symmetric} criteria, expressed in terms of $p$ and $d$,
where $p = \max_{A\in\mathcal{X}} \Pr(A)$
and $d$ is the maximum degree in $G_{\mathcal{X}}$.  A standard version of the LLL states that
if $ep(d+1) < 1$, then $\Pr(\bigcap \overline{A}) > 0$.  Given that all bad events \emph{can} be avoided, it is often
desirable to constructively find a point in the probability space (i.e., an assignment to variables in $\mathcal{V}$) that avoids them.
This problem has been thoroughly investigated in the sequential context~\cite{MoserT10,HarrisS14,Harris16,KolipakaS11,Kolmogorov16,HarveyV15,AchlioptasI14},
but somewhat less so from the point of view of parallel and distributed computation~\cite{ChungPS17,Ghaffari16,BrandtEtal16,ChangKP16,HaeuplerH17}.

The \emph{distributed} constructive LLL problem is the following.
The communications network is precisely $G_{\mathcal{X}}$.  Each vertex (event) $A$ knows the number of bad events in $G_{\mathcal{X}}$
and the distribution of those variables appearing in $\vbl(A) \subset \mathcal{V}$.  Vertices communicate for some number
of rounds, and collectively reach a consensus on an assignment to $\mathcal{V}$ in which no bad event occurs.
Moser and Tardos's~\cite{MoserT10} parallel resampling algorithm implies an $O(\log^2 n)$ time $\RandLOCAL$ algorithm
under the LLL criterion $ep(d+1) < 1$.  Chung, Pettie, and Su~\cite{ChungPS17} gave an $O(\log_{1/epd^2} n)$ time algorithm
under the LLL criterion $epd^2 < 1$ and an $O(\log n/\log\log n)$ time algorithm under criterion $p\cdot \poly(d)2^d = O(1)$.
They observed that under \emph{any} criterion of the form $p\cdot f(d) < 1$, $\Omega(\log^* n)$ time is necessary.
Ghaffari's~\cite{Ghaffari16} weak MIS algorithm, together with~\cite{ChungPS17}, implies
an $O(\log d \cdot \log_{1/ep(d+1)} n)$ algorithm under LLL criterion $ep(d+1)<1$.
Brandt et al.~\cite{BrandtEtal16} proved that $\Omega(\log_d \log n)$ time in $\RandLOCAL$ is necessary,
even under the permissive LLL criterion $p2^d \le 1$.   Chang et al.~\cite{ChangKP16}'s results imply that $\Omega(\log_d n)$
time is necessary in $\DetLOCAL$; however, there are no known deterministic distributed LLL algorithms.
It is conceivable that the distributed complexity of the LLL is very sensitive to the criterion used.
We {\em define} $T_{LLL}(n,d)$ to be the $\RandLOCAL$ time to compute a point in the probability space
avoiding all bad events (w.h.p.), under any ``polynomial'' LLL criterion of the form
\begin{equation}\label{eqn:LLL}
pd^c = O(1),
\end{equation}
where $c$ can be an arbitrarily large constant.
Prior results~\cite{ChungPS17,BrandtEtal16}
imply that $T_{LLL}(n,d)$ is $\Omega(\log_d\log n)$, $\Omega(\log^* n)$, and $O(\log_{1/epd^2} n)$.
In this section we prove an automatic speedup theorem for sublogarithmic $\RandLOCAL$ algorithms.
We do \emph{not} assume that $\Delta = O(1)$ in this section.

\begin{theorem}\label{thm:LLLspeedup}
Suppose that $\Alg$ is a $\RandLOCAL$ algorithm that solves some LCL problem $\mathcal{P}$ (w.h.p.),
in $T_\Delta(n)$ time.  For a sufficiently small constant $\epsilon>0$, suppose $T_\Delta(n)$
is upper bounded by $C(\Delta) + \epsilon \log_\Delta n$, for some function $C$.
It is possible to transform $\Alg$ into a new $\RandLOCAL$ algorithm $\Alg'$
that solves $\mathcal{P}$ (w.h.p.) in $O(C(\Delta)\cdot T_{LLL}(n,\Delta^{O(C(\Delta))}))$ time.
\end{theorem}

\begin{proof}
Suppose that $\Alg$ has a local probability of failure $1/n$,
that is, for any $v\in V(G)$, the probability that $N^r(v)$ is inconsistent
with $\mathcal{P}$ is $1/n$, where $r$ is the radius of $\mathcal{P}$.
Once we settle on the LLL criterion exponent $c$ in (\ref{eqn:LLL}), we fix $\epsilon = O((2c)^{-1})$.
Define $n^\star$ as the minimum value for which
\[
t^\star = T_\Delta(n^\star) < (1/2c)\cdot \log_\Delta n^\star - r.
\]
It follows that $t^\star = O(C(\Delta))$ and $n^\star = \Delta^{O(C(\Delta))}$.

The algorithm $\Alg'$ applied to an $n$-vertex graph $G$ works as follows.
Imagine an experiment where we run $\Alg$, but lie to the vertices, telling them that ``$n$'' = $n^\star$.
Any $v\in V(G)$ will see a $t^\star$-neighborhood $N^{t^\star}(v)$ that is consistent with some $n^\star$-vertex graph.
However, the \emph{bad event} that $N^r(v)$ is incorrectly labeled is $1/n^\star$, not $1/\poly(n)$, as desired.
We now show that this system of bad events satisfies the LLL criterion~(\ref{eqn:LLL}).
Define the following events, graph, and quantities:
\begin{align*}
\mathcal{E}_v &\::\: \mbox{the event that $N^r(v)$ is incorrectly labeled according to $\mathcal{P}$}\\
\mathcal{X} &= \{\mathcal{E}_v \;|\; v\in V(G)\}   & \mbox{the set of bad events}\\
G_{\mathcal{X}} &= (\mathcal{X}, \, \{(\mathcal{E}_u,\mathcal{E}_v) \;|\; N^{r+t^\star}(u) \cap N^{r+t^\star}(v) \neq \emptyset\}) & \mbox{the dependency graph}\\
d &\le \Delta^{2(r+t^\star)}\\
p &= 1/n^\star
\end{align*}
The event $\mathcal{E}_v$ is determined by the labeling of $N^r(v)$ and the label of each $v' \in N^r(v)$
is determined by $N^{t^\star}(v')$, hence $\mathcal{E}_v$ is determined by (the data stored in, and random bits generated by) vertices in $N^{r+t^\star}(v)$.
Clearly $\mathcal{E}_v$ is independent of any $\mathcal{E}_u$ for which $N^{r + t^\star}(u)\cap N^{r + t^\star}(v) = \emptyset$,
which justifies the definition of the edge set of $G_{\mathcal{X}}$.
Since the maximum degree in $G$ is $\Delta$, the maximum degree
$d$ in $G_{\mathcal{X}}$ is less than $\Delta^{2(r+t^\star)}$.
By definition of $\Alg$, $\Pr(\mathcal{E}_v) \le 1/n^\star = p$.
This system satisfies LLL criterion (\ref{eqn:LLL}) since, by definition of $t^\star$,
\[
pd^c = p\Delta^{2c(r+t^\star)} < (1/n^\star)\cdot n^\star = 1.
\]
The algorithm $\Alg'$ now simulates a constructive LLL algorithm on $G_{\mathcal{X}}$ in order to find a
labeling such that no bad event occurs.
Since a virtual edge $(\mathcal{E}_u,\mathcal{E}_v)$ exists iff $u$ and $v$ are at distance at most
$2(r+t^\star) = O(C(\Delta))$, any $\RandLOCAL$
algorithm in $G_{\mathcal{X}}$ can be simulated in $G$ with $O(C(\Delta))$ slowdown.
Thus, $\Alg'$ runs in $O(C(\Delta)\cdot T_{LLL}(n,\Delta^{O(C(\Delta))}))$ time.
\end{proof}

Theorem~\ref{thm:LLLspeedup} shows that when $\Delta=O(1)$, $o(\log n)$-time
$\RandLOCAL$ algorithms can be \emph{sped up} to run in $O(T_{LLL}(n, O(1)))$ time.
Another consequence of this same technique is that sublogarithmic $\RandLOCAL$ algorithms with \emph{large messages} can
be converted to (possibly slightly slower) algorithms with small messages.
The statement of Theorem~\ref{thm:LLLmessagesize}
reflects the use of a particular distributed LLL algorithm, namely~\cite[Corollary~1 and Algorithm~2]{ChungPS17}.
It may be improvable using future distributed LLL technology.

The LLL algorithm of~\cite{ChungPS17} works under the assumption that $epd^2 < 1$, and that 
each bad event $A \in \mathcal{X}$ is associated with a unique ID. 
The algorithm starts with a random assignment to the variables $\mathcal{V}$. In each iteration, let $\mathcal{F}$ be the set of bad events that occur under the current variable assignment; let $\mathcal{I}$ be the subset of $\mathcal{F}$ such that $A \in \mathcal{I}$ if and only if $\ID(A) < \ID(B)$ for each $B \in \mathcal{F}$ such that $\vbl(A) \cap \vbl(B) \neq \emptyset$. 
The next variable assignment is obtained by \emph{resampling} all variables in 
$\bigcup_{A \in \mathcal{I}} \vbl(A)$.  After $O(\log_{1/epd^2} n)$ iterations, no bad event occurs with probability $1-1/\poly(n)$.

\begin{theorem}\label{thm:LLLmessagesize}
Let $\Alg$ be a $(C(\Delta)+\epsilon \log_\Delta n)$-time $\RandLOCAL$ algorithm that solves some
LCL problem $\mathcal{P}$ with high probability, where $\epsilon>0$ is a sufficiently small constant.
Each vertex locally generates $r_\Delta(n)$ random bits and sends $m_\Delta(n)$-bit messages.
It is possible to transform $\Alg$ into a new $\RandLOCAL$ algorithm $\Alg'$
that solves $\mathcal{P}$ (w.h.p.) in
$O(\log_\Delta n)$ time, where each
vertex generates
$O(\log n + r_\Delta(\zeta)\cdot\log_{\zeta} n)$ random bits, and
sends
$O(\min\{\log (|\LabelOut|)\cdot \Delta^{O(1)} + m_\Delta(\zeta) + \zeta,\;\, r_\Delta(\zeta) \cdot \zeta\})$-bit messages,
where $\zeta = \Delta^{O(C(\Delta))}$ depends on $\Delta$.
\end{theorem}

\begin{proof}
We continue to use the notation and definitions from Theorem~\ref{thm:LLLspeedup}, and fix $c=3$ in the LLL criterion~(\ref{eqn:LLL}).
Since $d = \Omega(\Delta^{O(C(\Delta))}) = \Omega(\zeta)$ and we selected
$t^\star$ w.r.t. $c=3$ (i.e., LLL criterion $pd^3 < 1$), we have $1/epd^2 = \Omega(\zeta)$.
If $\Alg'$ uses the LLL algorithm of~\cite{ChungPS17}, each vertex $v \in V(G)$ will first generate an $O(\log n)$-bit unique identifier $\ID(\mathcal{E}_v)$ (which costs $O(\log n)$ random bits)
and generate $r_\Delta(n^\star)\cdot O(\log_{1/epd^2} n) = O(r_\Delta(\zeta)\cdot\log_{\zeta} n)$
random bits throughout the computation. Thus, the total number of random bits per vertex is $O(\log n + r_\Delta(\zeta)\cdot\log_{\zeta} n)$.

In each resampling step of $\Alg'$, in order for $v$ to tell whether 
$\mathcal{E}_v \in \mathcal{I}$, it needs the following information:
(i) $\ID(\mathcal{E}_u)$ for all $u \in N^{2(r+t^\star)}(v)$, and 
(ii) whether $\mathcal{E}_u$ occurs under the current variable assignment, for all $u \in N^{2(r+t^\star)}(v)$.
We now present two methods to execute one resampling step of $\Alg'$; they both take 
$O(C(\Delta))$ time using a message size that depends on $\Delta$ but is independent of $n$.
There are $O(\log_{1/epd^2} n) = O(\log_\zeta n) = O(\frac{\log_\Delta n}{C(\Delta)})$ resampling steps,
so the total time is $O(\log_\Delta n)$, independent of the function $C$.

\paragraph{Method 1.} Before the LLL algorithm proper begins, we do the following preprocessing step.
Each vertex $v$
gathers up all IDs and random bits in its $3(t^\star+r)$-neighborhood.
This takes $O((\log n + r_\Delta(\zeta)\cdot \log_{\zeta} n) \cdot \zeta / b)$
time with $b$-bit messages (recall that $\Delta^{O(t^\star+r)} = \Delta^{O(C(\Delta))} = \zeta$).
In particular, the runtime can be made $O(\log_\Delta n)$ if we set $b = O(r_\Delta(\zeta)\cdot \zeta)$.

During the LLL algorithm, each vertex $u$ owns one random variable:
an $r_\Delta(n^\star)$-bit string $V_u$. 
In order for $v$ to tell whether $\mathcal{E}_u$
occurs for each  $u \in N^{2(r+t^\star)}(v)$ under the current variable assignment, it only needs to know how many times
each $V_u$, $u\in N^{3(r+t^\star)}(v)$, has been resampled. 
Whether the output labeling of $u \in N^{2(r+t^\star)}(v)$ is locally consistent depends on the output labeling of vertices in $N^r(u)$, which depends on the random bits and the graph topology within $N^{r+t^\star}(u) \subseteq N^{3(r+t^\star)}(v)$. Given the graph topology, IDs, and the random bits within $N^{3(r+t^\star)}(v)$, the vertex $v$ can locally simulate $\Alg$ and decides whether  $\mathcal{E}_v \in \mathcal{I}$.

Thus, in each iteration of the LLL algorithm, each vertex $v$ simply needs to alert its
$3(r+t^\star)$-neighborhood whether $V_v$ is resampled or not.  This can be accomplished
in $O(r+t^\star) = O(C(\Delta))$ time with $\zeta$-bit messages.

\paragraph{Method 2.} In the second method, vertices keep their random bits private. Similar to the first method, we do a preprocessing step to let each vertex gathers up all IDs in its $2(t^\star+r)$-neighborhood.
This can be done in $O(\log_\Delta n)$ time using $\zeta$-bit messages.

During the LLL algorithm, in order to tell which subset
of bad events $\{\mathcal{E}_v\}_{v \in V(G)}$ occur under the current variable assignment,
all vertices simulate $\Alg$ for $t^\star$ rounds,
sending $m_\Delta(n^\star)$-bit messages. After the simulation,
for a vertex $v$ to tell whether $\mathcal{E}_v$ occurs, it needs to gather the output labeling of the vertices in $N^{r}(v)$.
This can be done in $r = O(1)$ rounds, sending $\log (|\LabelOut|) \cdot \Delta^{O(1)}$-bit messages.\footnote{An output label can be encoded as a $\log (|\LabelOut|)$-bit string.  We do not assume that $\Delta$ is constant so
$|\LabelOut|$, which may depend on $\Delta$ but not directly on $n$, is also not constant.  E.g., consider
the $O(\Delta)$ vertex coloring problem.}
Next, for a vertex $v$ to tell whether $\mathcal{E}_v \in \mathcal{I}$, it needs to know  whether $\mathcal{E}_u$ occurs for all $u \in N^{2(r+t^\star)}(v)$. This information can be gathered in $O(C(\Delta))$ time using messages of size $O(\zeta)$.
To summarize, the required message size is $O(\log (|\LabelOut|)\cdot \Delta^{O(1)} + m_\Delta(\zeta) + \zeta)$.
\end{proof}

An interesting corollary of Theorem~\ref{thm:LLLmessagesize}
is that when $\Delta=O(1)$,
randomized algorithms with unbounded length messages can be simulated with 1-bit messages.
\begin{corollary}
Let $\mathcal{P}$ be any LCL problem.
When $\Delta=O(1)$, any $o(\log n)$ algorithm solving $\mathcal{P}$ w.h.p.
using \emph{unbounded length messages}
can be made to run in $O(\log n)$ time with 1-bit messages.
\end{corollary}

\section{Conclusion}\label{sect:conclusion}

We now have a very good understanding of the $\LOCAL$ complexity landscape for cycles, tori,
bounded degree trees, and to a lesser extent, general bounded degree graphs.  
See Figure~\ref{fig:complexity}.  However, there are some very critical gaps in our understanding.

Our randomized speedup theorem of Section~\ref{sect:LLL} depends on the complexity
of a relatively weak version of the \Lovasz{} local lemma.  Since the LLL is essentially a ``complete''
problem for sublogarithmic $\RandLOCAL$ algorithms, understanding the distributed 
complexity of the LLL is a significant open problem.

\begin{conjecture}\label{conj:LLL}
There exists a sufficiently large constant $c$ such that the distributed LLL problem
can be solved in $O(\log\log n)$ time on bounded degree graphs, under the symmetric LLL criterion
$pd^c < 1$. 
\end{conjecture}

The new polynomial complexities introduced in Section~\ref{sec.poly} are of the form $\Theta(n^{1/k})$, 
$k\in\mathbb{Z}^+$.  Is this set of polynomial complexities exhaustive?  Is it possible to engineer problems with
complexity $\Theta(n^q)$ for any given rational $q$?  We think the answer is no, and resolving
Conjecture~\ref{conj:polygap} would be the first step.

\begin{conjecture}\label{conj:polygap}
Any $o(n)$-time $\DetLOCAL$ algorithm can be automatically sped up to run in $O(\sqrt{n})$ time.
In general, there is an $\omega(n^{1/(k+1)})$---$o(n^{1/k})$ gap in the $\DetLOCAL$ complexity hierarchy.
\end{conjecture}

One advantage of working with graph classes with bounded degree $\Delta$ is that arbitrary LCLs
can be encoded in $O(1)$ space by enumerating all acceptable configurations.  To study LCLs on
unbounded-degree graphs it is probably necessary to work with logical representations of LCLs,
and here the expressive power of the logic may introduce a new measure of problem complexity.
For example, the sentence ``$I$ is an MIS'' can be expressed using the predicate $[v\in I]$ and 
quantification over all vertices $v\in V(G)$, and all vertices $u\in N(v)$.

\bibliographystyle{plain}
\bibliography{../../references}

\newpage
\appendix

\section{Speedup Implications of Naor \& Stockmeyer}\label{sect:NaorStockmeyer}

Let $\mathcal{A}$ be any $T(n)$-round $\DetLOCAL$ algorithm.
Let $\eta$ and $\eta'$ be any two {\em order-indistinguishable} assignments of distinct IDs to
$N^{T(n)}(v)$,
i.e., for $u,w \in N^{T(n)}(v)$, $\eta(u) > \eta(w)$ if and only if $\eta'(u) > \eta'(w)$.
If, for every possible input graph fragment induced by $N^{T(n)}(v)$,
the output label of $v$ is identical under every pair of order-indistinguishable $\eta,\eta'$,
then $\mathcal{A}$ is \emph{order-invariant}.

Suppose that there exists a number $n'=O(1)$ such that $\Delta^{T(n') + r} \leq n'$.
If $\mathcal{A}$ is order-invariant then it can be turned into an $O(1)$-round
$\DetLOCAL$ algorithm $\mathcal{A}'$, since we can pretend that the total number
of vertices is $n'$ instead of $n$.

Naor and Stockmeyer~\cite{NaorS95} proved that any $\DetLOCAL$ algorithm that takes $\tau=O(1)$ rounds on a bounded degree graph can be turned into an order-invariant $\tau$-round $\DetLOCAL$ algorithm. A more careful analysis shows that the proof still works when $\tau$ is a slowly growing function of $n$.

\paragraph{Requirement for Automatic Speedup.}  The muticolor hypergraph Ramsey number $R(p,m,c)$ is the minimum number such that the following holds. Let $H$ be a complete $p$-uniform hypergraph of at least $R(p,m,c)$ vertices. Then any $c$-edge-coloring of $H$ contains a monochromatic clique of size $m$.

Given the number $\tau\ge 2$, the three parameters $p$, $m$, and $c$ are selected as follows. (See the proof of~\cite[Lemma 3.2]{NaorS95} for more details.)
\begin{itemize}
\item The number $p$ is the maximum number of vertices in $N^\tau(v)$, over all vertices $v \in V(G)$ and all graphs $G$ under consideration. For rings, $p = 2\tau+1$. For tori, $p \leq 2(\tau+1)^2$. For trees or general graphs,
$p \leq \Delta^{\tau}$.
\item The number $m$ is the maximum number of vertices in $N^{\tau+r}(v)$, over all vertices $v \in V(G)$ and all graphs $G$ under consideration. E.g., for rings, $p = 2\tau+2r+1$ and for general graphs, $p \leq \Delta^{\tau+r}$.
\item The number $z$ counts the distinguishable radius-$\tau$ centered subgraphs, disregarding IDs.  For example, for LCLs on the ring without input labels or port numbering, $z=1$, whereas with input labels and port numbering
it is $(2|\LabelIn|)^{2\tau+1}$ since each vertex has one of $|\LabelIn|$ input labels and 2 port numberings.
In general $z$ is less than $2^{{\Delta^\tau}\choose 2} \cdot (\Delta! |\LabelIn|)^{p}$.
\item The number $c$ is defined as $|\LabelOut|^{p! z}$.
Intuitively, we can use a number in $[c]$ to encode a function that maps a radius-$\tau$ centered subgraph (that is equipped with unique vertex IDs from a set $S$ with cardinality $p$) to an output label in $\LabelOut$.
\end{itemize}

Recall that vertices in $\DetLOCAL$ have $O(\log n)$-bit IDs, i.e.,
they can be viewed as elements of $[n^k]$ for some $k=O(1)$.
Naor and Stockmeyer's proof implies that, as long as $n^k \geq R(p,m,c)$, any $\DetLOCAL$ $\tau$-round algorithm on a bounded degree graph can be turned into an order-invariant $\tau$-round $\DetLOCAL$ algorithm,
which then implies an $O(1)$-round $\DetLOCAL$ algorithm.

\paragraph{The Ramsey number $R(p,m,c)$.} According to the proof of~\cite[\S 1, Theorem 2]{GrahamRS90}, we have:
\begin{align*}
\mbox{For } p = 1, && R(p,m,c) &= c(m-1) + 1 &&\\
\mbox{For } p > 1, && R(p,m,c) &\leq 2c^x &&\text{ where } x = \sum_{i=p-1}^{R(p-1,m,c)-1} {i+1 \choose p-1} < R(p-1,m,c)^{p}
\end{align*}
Therefore, $\log^\ast(R(p,m,c)) \leq p + \log^\ast m + \log^\ast c + O(1)$.

\paragraph{Automatic Speedup Theorems.}
Observe that in all scenarios, if the running time $\tau=\tau(n) = \omega(1)$,
we have $\log^\ast m + \log^\ast c = o(p)$.
Therefore, having $p \leq \epsilon \log^\ast n$ for some small enough constant $\epsilon$ suffices to meet the condition  $n^k \geq R(p,m,c)$.
We conclude that the complexity of any LCL problem (with or without input labels and port numbering)
in the $\LOCAL$ model never falls in the following gaps:
\begin{align*}
&\text{$\omega(1)$---$o(\log^\ast n)$} && \mbox{ for } \text{$n$-rings}.\\
&\text{$\omega(1)$---$o(\sqrt{\log^\ast n})$} && \mbox{ for } \text{$(\sqrt{n}\times \sqrt{n})$-tori}.\\
&\text{$\omega(1)$---$o(\log(\log^\ast n))$}  && \mbox{ for bounded degree trees or general graphs}.
\end{align*}
Due to the ``Stepping-Up Lemma'' (see~\cite[\S 4, Lemma 17]{GrahamRS90}), we have $\log^\ast(R(p,m,2)) = \Omega(p)$. Thus, Naor and Stockmeyer's approach \emph{alone} cannot give
an $\omega(1)$---$o(\log^\ast n)$ gap for general graphs.

However, for a certain class of LCL problems on
$(\sqrt{n}\times \sqrt{n})$-tori, the gap can be widened to $\omega(1)$---$o(\log^\ast n)$~\cite[p. 2]{Brandt+17}.  The following proof is due to Jukka Suomela (personal communication).

\begin{theorem}[J. Suomela]\label{thm:torigap}
Let $\mathcal{P}$ be any LCL problem on $(\sqrt{n}\times \sqrt{n})$-tori that does not refer to input labels or
port-numbering. The $\DetLOCAL$
and $\RandLOCAL$ complexity of $\mathcal{P}$ is either  $O(1)$ or $\Omega(\log^\ast n)$.
\end{theorem}
\begin{proof}
Given a $(\sqrt{n}\times \sqrt{n})$-torus $G$, we associate each vertex $v \in V(G)$  with a coordinate $(\alpha,\beta)$, where $\alpha, \beta \in \{0, \ldots, \sqrt{n}-1\}$. We consider the following special way to generate unique $k\log n$-bit IDs. Let $\phi_x$ and $\phi_y$ be two functions mapping integers in $\{0, \ldots, \sqrt{n}-1\}$ to integers in $\{0, \ldots, n^{k/2} - 1\}$.
We additionally require that $\phi_x(0) < \ldots < \phi_x(\sqrt{n}-1) < \phi_y(0) < \ldots < \phi_y(\sqrt{n}-1)$.
If $v$ is at position $(\alpha,\beta)$, it has ID $\phi_x(\alpha) \cdot n^{k/2}  + \phi_y(\beta)$.
Notice that the IDs of all vertices in $N^\tau(v)$ can be deduced from just $4\tau+2$ numbers:
$\phi_x(i)$, $i \in [\alpha-\tau,\alpha+\tau]$ and $\phi_y(j)$, $j \in [\beta-\tau,\beta+\tau]$.

Suppose that the complexity of $\mathcal{P}$ is $o(\log^\ast n)$.
Let $\mathcal{A}$ be any $\tau$-round $\DetLOCAL$ algorithm for solving $\mathcal{P}$,
where $\tau = o(\log^\ast n)$.
Notice that the algorithm $\mathcal{A}$ works correctly even when we restrict ourselves to the above special ID assignment.
Our goal is to show that $\mathcal{P}$ is actually {\em trivial} in the sense that there exists an element $\sigma \in \LabelOut$ such that labeling all vertices by $\sigma$ gives a legal labeling, assuming w.l.o.g.~that $\sqrt{n} > 2r+1$. Thus, $\mathcal{P}$ can be solved in $O(1)$ rounds.

In subsequent discussion, we let  $v$ be any vertex whose position is $(\alpha,\beta)$, where $\tau+r \leq \alpha \leq (\sqrt{n}-1) - (\tau+r)$ and $\tau+r \leq \beta \leq (\sqrt{n}-1) - (\tau + r)$. That is, $v$ is sufficiently far way from the places where the coordinates wrap around.

Given  $\mathcal{A}$, we construct a function $f$ as follows.
Let $S=(s_1, \ldots, s_{4\tau+2})$ be a vector of $4\tau+2$ numbers in $\{0, \ldots, n^{k/2}-1\}$ such that $s_k < s_{k+1}$ for each $k \in [4\tau +2]$.
Then $f(S) \in \LabelOut$ is defined as the output labeling of $v$  resulting from executing $\mathcal{A}$ with the following ID assignment of vertices in $N^{\tau}(v)$.
We set $\phi_x(\alpha-\tau-1 + i) = s_i$ for each $i \in [2\tau + 1]$ and set $\phi_y(\beta-\tau-1 + j) = s_{j + 2\tau + 1}$ for each  $j \in [2\tau + 1]$ 
Recall that $\mathcal{P}$ does not use port-numbering and input labeling,
so the output labeling of $v$ depends only on IDs of vertices in $N^{\tau}(v)$.

We set $p = 4\tau+2$, $m = 4\tau+4r+2$, and $c = |\LabelOut|$. 
Notice that the calculation of the parameter $c$ here is different from the original proof of Naor and Stockmeyer.
Since we already force that $\phi_x(0) < \ldots < \phi_x(\sqrt{n}-1) < \phi_y(0) < \ldots < \phi_y(\sqrt{n}-1)$, we do not need to consider all $p!$ permutations of the set $S$.

We have $R(p,m,c) \ll n^{k/2}$ (since $p = o(\log^\ast n)$).
Thus, there exists a set $S'$ of $m$ distinct numbers in $\{0, \ldots, n^{k/2}\}$ such that the following is true.  We label these $m$ numbers $\phi_x(i)$, $i \in [\alpha-\tau-r,\alpha+\tau+r]$, and $\phi_y(j)$, $j \in [\beta-\tau-r,\beta+\tau+r]$ by the set $S'$ such that $\phi_x(\alpha-\tau-r) < \ldots < \phi_x(\alpha+\tau+r) < \phi_y(\beta-\tau-r) < \ldots < \phi_y(\beta+\tau+r)$. Then the output labels of all vertices in $N^r(v)$ assigned by $\mathcal{A}$ are identical.

Therefore, there exists an element $\sigma \in \LabelOut$ such that labeling all vertices by $\sigma$ yields a legal labeling of $G$.
Thus, $\mathcal{P}$ can be solved in $O(1)$ rounds.
\end{proof}


\paragraph{Discussion.}
It still remains an outstanding open problem whether the gap for other cases can also be widened to $\omega(1)$---$o(\log^\ast n)$.

The proof of Theorem~\ref{thm:torigap} extends easily to $d$-dimensional tori, but
does not extend to bounded degree trees, since
there is a \underline{non-trivial} problem that can be solved in $O(1)$ rounds on a subset of bounded degree trees.
Naor and Stockmeyer~\cite{NaorS95} showed that on any graph class in which all vertex degrees are odd,
\emph{weak $2^{O(\Delta\log\Delta)}$-coloring} can be solved in $2$ rounds and \emph{weak $2$-coloring}
can be solved in $O(\log^\ast \Delta)$ rounds in $\DetLOCAL$.\footnote{A weak coloring is one in which every vertex is colored differently than at least one neighbor.}
This problem is \emph{non-trivial} in the sense that coloring all vertices by the same color is not a legal solution.
Since the $d$-dimensional torus is $\Delta$-regular, $\Delta=2d$,
we conclude that the complexity of weak $O(1)$-coloring on $\Delta$-regular graphs is
$\Theta(\log^\ast n)$ for every fixed even number $\Delta\ge 2$.

Theorem~\ref{thm:torigap} also does not extend to LCL problems that use input labels or port-numbering. If either input labels or port-numbering are allowed, then one can construct a non-trivial LCL problem that can be solved in $O(1)$ rounds even on cycle graphs. An {\em orientation} of a vertex $v \in V(G)$ is defined as a port-number in $[\deg(v)]$, indicating a vertex in $N(v)$ that $v$ is pointed towards.
An {\em $\ell$-orientation}
of a cycle $G$ is an orientation of all vertices in $G$ meeting the following conditions.
If $|V(G)| \leq \ell$, then all vertices in $G$ are oriented to the same direction, i.e., no two vertices point toward each other.
If $|V(G)| > \ell$, then each vertex $v \in V(G)$ belongs to a path $P$ such that
(i) all vertices in $P$ are oriented to the same direction (no two point to each other), and
(ii) the number of vertices in $P$ is at least $\ell$.
Notice that $\ell$-orientation, $\ell=O(1)$, is an LCL that refers to port-numbering.
We show that in $O(1)$ rounds we can compute an $\ell$-orientation of $G$ for any constant $\ell$.\footnote{Even though orienting all vertices in the cycle to the same direction gives a legal labeling, $\ell$-orientation is still a non-trivial LCL problem. Consider a subpath $(v_1, v_2, v_3, v_4)$ in the cycle. Suppose that the port-number of $(v_2,v_3)$ stored at $v_2$ is 1, but the port-number of $(v_3,v_4)$ stored at $v_3$ is 2.
Then we need to label $v_2$ and $v_3$ differently (1 and 2, respectively) in order to orient them in the same direction `$\rightarrow$'.}

\begin{theorem}\label{lem:ori}
Let $G$ be a cycle graph and $\ell$ be a constant.
There is a $\DetLOCAL$ algorithm that computes an $\ell$-orientation of $G$ in $O(1)$ rounds.
\end{theorem}
\begin{proof}
This is a known result. See~\cite[Fact 5.2]{HanckowiakKP01}
or~\cite[Lemma 3.6, Case B]{FischerG17} for a sketch of the proof.
For the sake of completeness, we present a full proof.
We first show how to compute a $2$-orientation of a cycle $G$ in $O(1)$ rounds,
and then we extend it to any constant $\ell$.

\paragraph{Computing a $2$-orientation.}
We assume $|V(G)| \geq 3$.
A $\DetLOCAL$ $O(1)$-round algorithm to compute a $2$-orientation is described as follows.
First, each vertex $v \in V(G)$ computes an arbitrary orientation.
With respect to this orientation of $G$, define sets $V_1,V_2,V_3$ as follows.
\begin{itemize}
\item $v \in V_1$ if and only if there exists $u\in N(v)$ such that $u$ and $v$ are oriented to the same direction.
\item $v \in V_2$ if and only if there exists $u\in N(v) \setminus V_1$ such that $u$ and $v$ are oriented toward each other.
\item $V_3 = V(G) \setminus (V_1 \cup V_2)$. Observe that for each $v \in V_3$, there exists $u \in N(v) \cap V_1$.
\end{itemize}
A 2-orientation is obtained by re-orienting the vertices in $V_2$ and $V_3$. The vertices in $V_2$ are partitioned into unordered pairs such that $u, v \in V_2$ are paired-up if and only if (i) $\{u,v\} \in E(G)$ and (ii) $u$ and $v$ are oriented toward each other. For each pair $\{u,v\}$, reverse the orientation of any one of $\{u,v\}$.  For each vertex $v \in V_3$, let $u$ be any neighbor of $v$ such that $u \in V_1$, and re-orient $v$ to the orientation of $u$.

\paragraph{Computing an $\ell$-orientation.}
For each positive integer $k$, we recursively define a $\DetLOCAL$ $O(1)$-round algorithm $\mathcal{A}_k$ which computes a $k$-orientation.
In what follows, we assume $k \geq 3$ and $|V(G)| \geq k$.

The algorithm  $\mathcal{A}_k$ is described as follows.
First, execute $\mathcal{A}_{\lceil k/2 \rceil}$ to obtain a $\lceil k/2 \rceil$-orientation of $G$.
With respect to this orientation of $G$, define the following terminologies.
Let $\Pset$ be the set of all maximal-size connected subgraphs  in $G$ such that all constituent vertices are oriented to the same direction.
Notice that if $\Pset$ contains  a cycle, then $\Pset = \{G\}$.
Otherwise $\Pset$ contains only paths.
Define $\Pset_1$ as the subset of $\Pset$ such that $P \in \Pset_1$ if and only if the number of vertices in $P$ is at least $k$.
Define $\Pset_2$ as the subset of $\Pset \setminus \Pset_1$ such that $P \in \Pset_2$ if and only if there exists another path $P' \in \Pset \setminus \Pset_1$ meeting the following condition. There exist an endpoint $u$ of $P$ and an endpoint $v$ of $P'$ such that $\{u,v\} \in E(G)$, and $u$ and $v$ are oriented toward each other.
Define $\Pset_3 = \Pset \setminus (\Pset_1 \cup \Pset_2)$.
Observe that each $P \in \Pset_3$ is adjacent to a path in $\Pset_1$.

The paths in $\Pset_2$ are partitioned into unordered pairs such that $P, P' \in \Pset_2$ are paired-up if and only if there exist an endpoint $u$ of $P$ and an endpoint $v$ of $P'$ such that $\{u,v\} \in E(G)$, and $u$ and $v$ are oriented toward each other. For each pair $\{P,P'\}$, reverse the orientation of all the vertices in any one of $\{P,P'\}$.   For each path $P \in \Pset_3$, let $P' \in  \Pset_1$ be any path adjacent to $P$, and re-orient $P$ to the orientation of $P'$.

The round complexity of $\mathcal{A}_k$ satisfies the recurrence $T(k) = T(\ceil{k/2}) + O(k)$, which is $O(k)$.
\end{proof}

Notice that if the $(\sqrt{n}\times \sqrt{n})$-torus is {\em oriented} (i.e., the input port-numberings all agree with a fixed N/S/E/W orientation~\cite{Brandt+17}), then there is no non-trivial LCL problem solvable in $O(1)$ time.






\end{document}